\RequirePackage[2025-11-01]{latexrelease}
\RemoveFromHook{package/amsthm/after}[firstaid/aliascounter]
\PassOptionsToPackage{dvipsnames}{xcolor}
\documentclass[a4paper,UKenglish,cleveref,autoref,thm-restate]{lipics-v2021}

\title{Tight (Double) Exponential Bounds for 
Identification Problems: Locating-Dominating Set and Test Cover} 

\titlerunning{Tight (Double) Exponential Bounds for 
Identification Problems} 

\author{Dipayan Chakraborty}{Université Clermont Auvergne, CNRS, Mines Saint-Étienne, Clermont Auvergne INP, LIMOS, 63000 Clermont-Ferrand, France\and Department of Mathematics and Applied Mathematics, University of Johannesburg, Auckland Park, 2006, South Africa \and
Department of Computer Science and Mathematics, Lebanese American University, Beirut, Lebanon \and
Aix-Marseille Universit\'e, CNRS, Ecole Centrale M\'editerran\'ee, LIS, Marseille, France \and
Department of Computer Science and Engineering, SRM University-AP, Amaravati 522 240, Andhra Pradesh, India \and
 \url{https://dipayan5186.github.io/Website/}}{dipayan.chakraborty@uca.fr}{https://orcid.org/0000-0001-7169-7288}{International Research Center ``Innovation Transportation and Production Systems'' of the I-SITE CAP 20-25.}

\author{Florent Foucaud}{Université Clermont Auvergne, CNRS, Mines Saint-Étienne, Clermont Auvergne INP, LIMOS, 63000 Clermont-Ferrand, France\and \url{https://perso.limos.fr/ffoucaud}}{florent.foucaud@uca.fr}{https://orcid.org/0000-0001-8198-693X}{ANR project GRALMECO (ANR-21-CE48-0004), French government IDEX-ISITE initiative 16-IDEX-0001 (CAP 20-25), International Research Center ``Innovation Transportation and Production Systems'' of the I-SITE CAP 20-25, CNRS IRL ReLaX.}

\author{Diptapriyo Majumdar}{Indraprastha Institute of Information Technology Delhi, New Delhi, India \and \url{https://diptapriyomajumdar.wixsite.com/toto}}{diptapriyo@iiitd.ac.in}{https://orcid.org/0000-0003-2677-4648}{Supported by Science and Engineering Research Board (SERB) grant SRG/2023/001592.}

\author{Prafullkumar Tale}{Indian Institute of Science Education and Research Pune, Pune, India \and \url{https://pptale.github.io/}}{prafullkumar@iiserpune.ac.in}{https://orcid.org/0000-0001-9753-0523}{Supported by INSPIRE Faculty Fellowship  DST/INSPIRE/04/2021/00314.}

\authorrunning{Chakraborty, Foucaud, Majumdar, and Tale} 

\Copyright{Chakraborty, Foucaud, Majumdar, Tale} 



\category{} 

\relatedversion{{An extended abstract of this article will appear in proceedings of ISAAC-2024.}} 

\EventEditors{John Q. Open and Joan R. Access}
\EventNoEds{2}
\EventLongTitle{42nd Conference on Very Important Topics (CVIT 2016)}
\EventShortTitle{CVIT 2016}
\EventAcronym{CVIT}
\EventYear{2016}
\EventDate{December 24--27, 2016}
\EventLocation{Little Whinging, United Kingdom}
\EventLogo{}
\SeriesVolume{42}
\ArticleNo{23}
\nolinenumbers
\usepackage[linesnumbered,algoruled,boxed,lined]{algorithm2e}
\usepackage{complexity} 
 
\usepackage{tikz}
\usetikzlibrary{positioning,backgrounds,patterns,calc}

\usepackage{enumitem}

\usepackage{amsmath}
\usepackage{amsthm}
\usepackage{thmtools}
\usepackage{thm-restate}

\usetikzlibrary{patterns}

\usepackage{subcaption}
\usepackage{enumitem}
 
 \usepackage{graphicx}

\usepackage[dvipsnames]{xcolor}

\newcommand{\ff}[1]{\textcolor{violet}{#1 - FF}}

\newcommand{\pt}[1]{\textcolor{purple}{#1 - PT}}

\newcommand{\LD}{\textsc{Locating-Dominating Set}\xspace}
\newcommand{\TCPB}{\textsc{Test Cover}\xspace}

\newcommand{\setrep}{\texttt{set-rep}}
\newcommand{\bitrep}{\texttt{bit-rep}}
\newcommand{\bit}{\texttt{bit}}

\newcommand{\dptw}{\text{\textbf{d}}}

\newcommand{\ETH}{{\textsf{ETH}}}

\newcommand{\tw}{\textsf{tw}}

\newcommand{\vc}{\textsf{vc}}

\newcommand{\calA}{\mathcal{A}}
\newcommand{\calB}{\mathcal{B}}

\newcommand{\calF}{\mathcal{F}}

\newcommand{\calO}{\mathcal{O}}
\newcommand{\calP}{\mathcal{P}}

\newcommand{\calT}{\mathcal{T}}

\newcommand{\calX}{\mathcal{X}}
\newcommand{\calY}{\mathcal{Y}}

\newcommand{\true}{\texttt{True}}
\newcommand{\false}{\texttt{False}}

\newcommand{\yes}{\textsc{Yes}}
\newcommand{\no}{\textsc{No}}

\newtheorem*{theorem*}{Theorem}

\newcommand{\defproblem}[3]{
  \vspace{1mm}
\noindent\fbox{
  \begin{minipage}{0.96\textwidth}
  \begin{tabular*}{\textwidth}{@{\extracolsep{\fill}}lr} #1 \\ \end{tabular*}
  {\bf{Input:}} #2  \\
  {\bf{Question:}} #3
  \end{minipage}
  }
  \vspace{1mm}
}

\begin{document}

\maketitle

\begin{abstract}
We investigate fine-grained algorithmic aspects for
classical 
identification problems
in graphs, namely \textsc{Locating-Dominating Set}, and 
in set systems, namely \textsc{Test Cover}.
In the first problem, 
an input is a graph $G$ on $n$ vertices 
and an integer $k$, and the objective is to decide whether there is 
a subset $S$ of $k$ vertices such that any two distinct vertices not in $S$ 
are dominated by distinct subsets of $S$. 
In the second problem, an input is a set $U$ of items,
a collection $\calF$ of subsets of $U$ called \emph{tests},
and an integer $k$, and 
the objective is to select a set $S$ of at most $k$ tests
such that any two distinct items are contained in
a distinct subset of tests of $S$. 
For our first result, we show that
\textsc{Locating-Dominating Set} (respectively, \textsc{Test Cover}) 
parameterized by the treewidth of the input graph 
(respectively, the incidence graph) does not admit an algorithm 
running in time  $2^{2^{o(\tw)}} \cdot {\rm poly}(n)$
(respectively, $2^{2^{o(\tw)}} \cdot {\rm poly}(|U| +  |\calF|)$),
unless the \ETH\ fails.
Next, we prove that
unless the \ETH\ fails,
\textsc{Locating-Dominating Set} does not 
admit an algorithm running in time $2^{o(k^2)} \cdot {\rm poly}(n)$,
nor a polynomial-time kernelization algorithm that reduces 
the solution size and outputs a kernel with $2^{o(k)}$ vertices, and
\textsc{Test Cover} does not 
admit an algorithm running in time 
$2^{2^{o(k)}} \cdot {\rm poly}(|U| + |\calF|)$ nor 
a kernel with $2^{2^{o(k)}}$ vertices.
Again, we show that these lower bounds are tight
by designing (kernelization) algorithms with matching running times.
\end{abstract}

\tableofcontents
\newpage

\section{Introduction}
\label{sec:intro}

The article aims to study the algorithmic properties of certain identification problems in discrete structures.
In identification problems, one wishes to select a solution substructure
of an input structure (a subset of vertices, the coloring of a graph, etc.)
so that the solution substructure uniquely identifies each element.
Some well-studied examples are, for example, the problems \TCPB\ for set systems and
\textsc{Metric Dimension} for graphs 
(Problems~[SP6] and [GT61] in the book by Garey and Johnson~\cite{GJ79}, respectively). 
This type of problem has been studied since the 1960s, both in combinatorics
community (see e.g. R'enyi~\cite{renyi1961} or
Bondy~\cite{bondy1972induced}), and in the algorithms
community since the 1970s~\cite{BDK05,BHHHLRS03,CGJSY12,MS85}.
They have multiple practical and theoretical applications, such as network
monitoring~\cite{Rao93}
, medical diagnosis~\cite{MS85},
bioinformatics~\cite{BDK05}, coin-weighing problems~\cite{sebo04},
graph isomorphism~\cite{B80}, games~\cite{Chvatal83}, machine
learning~\cite{CN98} etc. An online bibliography on the topic with
over 500 entries as of 2024 is maintained at~\cite{onlineBIBidproblems}.

In this article, we investigate fine-grained algorithmic aspects of two identification problems.
One of them is {\LD}, which is a graph-theoretic problem, and the other one {\TCPB} is a problem on set systems.
Like most other interesting and practically motivated computational
problems, identification problems also turned out to be
\NP-hard, even in very restricted settings.
See, for example, \cite{colbourn1987locating} and \cite{GJ79},
respectively.
We refer the reader to `Related Work’ towards the end of this section for a more detailed overview of their algorithmic complexity.

To cope with this hardness, these problems have been studied through the lens
of parameterized complexity.
In this paradigm, we associate each instance $I$ with a parameter
$\ell$, and are interested to know whether the problem
admits a \emph{fixed parameter tractable} (\FPT) algorithm,
i.e., an algorithm with the running time
$f(\ell) \cdot |I|^{\calO(1)}$, for some computable
function $f$.
A parameter may originate from the formulation of the problem itself or can be a property of the input graph.
If a parameter originates from the formulation of the problem itself, then it is called a \emph{natural parameter}.
Otherwise, a parameter that is described by the structural property of the input graph is called a \emph{structural parameter}.
One of the most well-studied structural parameters is
`treewidth’ (which, informally, quantifies how close
the input graph is to a tree, and is {denoted by ${\tw}$}).
We refer readers to \cite[Chapter $7$]{cygan2015parameterized} for a formal definition.
{Courcelle’s celebrated theorem~\cite{Courcelle90} states that there exists a computable function 
$f$ such that, given a graph $G$ of treewidth $\tw$ and an MSOL formula 
$\varphi$ of length $L$, one can decide whether \(G \models \varphi \) in time
\(f(\tw, L)\cdot n\), where $n = |V(G)|$.
In particular, for every fixed MSOL formula $\varphi$, the model-checking problem 
is fixed-parameter  tractable when parameterized by $\tw$.}
Hence, a large class of problems admits an \FPT\ 
algorithm when parameterized by the treewidth.
Unfortunately, the function $f$ is a tower of exponents whose height
depends roughly on the size of the MSOL formula.
Hence, this result serves as a starting point to
obtain an (usually impractical) \FPT\ algorithm.

Over the years, researchers have searched for more efficient problem-specific
algorithms when parameterized by the treewidth.
There is a rich collection of problems that admit an \FPT\ 
algorithm with single- or almost-single-exponential dependency
with respect to treewidth,
i.e., of the form $2^{\calO(\tw)} \cdot n^{\calO(1)}$ or
$2^{\calO(\tw \log (\tw))} \cdot n^{\calO(1)}$,
(see, for example,~\cite[Chapter~$7$]{cygan2015parameterized}).
There are a handful of graph problems that only admit \FPT\ algorithms with double- or triple-exponential dependence in the treewidth~\cite{BliznetsH-TAMC-24,DBLP:journals/ai/FichteHMTW23,DBLP:conf/sat/FichteHMW18,DBLP:conf/lics/FichteHP20,DBLP:journals/corr/abs-2402-14927,HKL24,MM16}.
In the respective articles, the authors prove that this double- (respectively,
triple-) dependence on the treewidth cannot be improved
unless the Exponential-Time Hypothesis (\ETH)\footnote{The \ETH\ roughly states that $n$-variable {\sc 3-SAT} cannot be solved in time $2^{o(n)}n^{\calO(1)}$.
See \cite[Chapter $14$]{cygan2015parameterized}.} fails.

{Foucaud et al.~\cite{FoucaudGK0IST24}} 
studied three {\NP-complete} \emph{metric-based graph problems}
viz \textsc{Metric Dimension}, \textsc{Strong Metric Dimension}, and \textsc{Geodetic Set}.
They proved that these problems
admit double-exponential lower bounds in $\tw$ (and, in fact, in the size of the minimum vertex cover size $\vc$ for the second problem) under the {\ETH}.
The first two of these three problems are identification problems.

In this article, we continue this line of research and highlight the usefulness of identification problems to prove relatively rare types of lower bounds, by investigating fine-grained algorithmic aspects of \LD and \TCPB, two classical (non-metric-based) identification problems.
This also shows that this type of bound can hold for ``local’’ (i.e., non-metric-based) problems (the problems studied in~\cite{FoucaudGK0IST24} were all metric-based).
Apart from serving as examples for double-exponential dependence on treewidth, these problems are of interest in their own right, and possess a rich literature both in the algorithms
and discrete mathematics communities, as highlighted in
`Related Work’.

\defproblem{\LD}{A graph $G$ on $n$ vertices and an integer $k$.}{Does
there exist a locating-dominating set of size $k$ in $G$, that is,
a set $S$ of $V(G)$ of size at most $k$ such that for any two different
vertices $u, v \in V(G) \setminus S$,
their neighborhoods in $S$ are different,
i.e., $N(u) \cap S \neq N(v) \cap S$ and non-empty?}

\defproblem{\TCPB}{A set of items $U$, a collection $\calF$ of subsets of $U$ called \emph{tests}, and an integer $k$.}{Does there exist a collection of at most $k$ tests
such that for each pair of items, there is a test
that contains exactly one of the two items?}

\noindent As \TCPB\ is defined over set systems,
for structural parameters, we define an \emph{incidence graph}
in the natural way:
A bipartite graph $G$ on $n$ vertices with bipartition $\langle R, B \rangle$
of $V(G)$ such that sets $R$ and $B$ contain a vertex
for every set in $\calF$ and for every item in $U$, respectively.
A vertex $r \in R$ and a vertex $b \in B$ are adjacent if and only if the set corresponding to $r$ contains
the element corresponding to $b$.
In this incidence graph, a test cover corresponds to a subset $S$ of $R$ such that any two vertices of $B$ have distinct neighborhoods within $S$.

The \LD\ problem is also a \emph{graph domination problem}.
In the classical \textsc{Dominating Set} problem, an input is an undirected graph $G$ and an integer $k$, and the objective is to decide whether there is a subset $S \subseteq V(G)$ of size $k$ such that for any vertex $u
\in V(G) \setminus S$, at least one of its neighbors is in $S$.
It can also be seen as a local version of \textsc{Metric Dimension}\footnote{Note that \textsc{Metric Dimension}
is also an identification
problem, but it is inherently non-local in nature, and indeed was
studied together with two other non-local problems
in~\cite{FoucaudGK0IST24}, where the similarities
between these non-local problems was noticed.}
in which the input is the same and the objective is to determine
a set $S$ of $V(G)$ such that for any two vertices $u, v \in V(G) \setminus S$,
there exists a vertex $s \in S$ such that
$dist(u, s) \neq dist(v, s)$.

We demonstrate the applicability of the techniques
from~\cite{FoucaudGK0IST24} to the ``local’’ problems {\LD} and {\TCPB}, showing that the metric nature of e.g. \textsc{Metric Dimension} was in fact not essential to obtain this type of lower bounds.
To do so, we adapt the main techniques developed in~\cite{FoucaudGK0IST24} to our setting, namely, the \emph{set-representation gadgets} and \emph{bit-representation gadgets}.

We prove the following result.

\begin{restatable}{theorem}{locdomtreewidthlb}
\label{thm:TW}
Unless the \ETH\ fails,
\LD (respectively, \TCPB)
does not admit an algorithm
running in time  $2^{2^{o(\tw)}} \cdot {\rm poly}(n)$,
where $\tw$ is the treewidth and $n$ is the order of the graph (respectively, of the incidence graph).
\end{restatable}

We also prove that both \LD and \TCPB  admit an algorithm
with matching running time (Theorem~\ref{thm:LD-tw-algo}
and Theorem~\ref{thm:TC-tw-algo}), by nontrivial dynamic programming schemes on tree-decompositions.

In contrast to Theorem~\ref{thm:TW}, \textsc{Dominating Set} admits an algorithm
running in time $\calO(3^{\tw} \cdot
n^{2})$~\cite{LokshtanovMS18,RooijBR09}.

We remark that the algorithmic lower bound of Theorem~\ref{thm:TW}
holds true even with respect to vertex integrity~\cite{DBLP:journals/tcs/GimaHKKO22}, a parameter larger than
treewidth. This also implies the same lower bounds for parameters treedepth and pathwidth, whose values both always lie between the vertex integrity and the treewidth.

Theorem~\ref{thm:TW} adds
\LD\ and \TCPB\ to the short list of {\NP}-complete problems that admit (tight)
double-exponential lower bounds for treewidth.
Using the techniques mentioned in~\cite{FoucaudGK0IST24},
two more problems (from learning theory), viz. \textsc{Non-Clashing Teaching Map} and \textsc{Non-Clashing Teaching Dimension},
were recently shown
in~\cite{DBLP:journals/corr/abs-2309-02876} to
admit similar lower bounds.

Next, we prove that \LD and \TCPB
also admit `exotic’ lower bounds,
when parameterized by the solution size $k$.
First, note that both problems are trivially \FPT\ 
when parameterized by the solution size.
Indeed, as any solution must have size at
least logarithmic in the number of elements/vertices (assuming no
redundancy in the input), the whole instance is a trivial
single-exponential kernel for \LD, and double-exponential in the case
of \TCPB.
To see this, note that in both problems, any two vertices/items must be assigned
a distinct subset from the solution set. Hence, if there are more than
$2^k$ of them, we can safely reject the instance. Thus, for \LD, we can
assume that the graph has at most $2^k+k$ vertices, and for \TCPB, at most $2^k$ items.
Moreover, for \TCPB, one can also assume that every test is unique (otherwise, delete any redundant test), in which case
there are at most ${2^2}^k$ tests.
Hence, \LD\ admits a kernel with size $\calO(2^{k})$,
and an \FPT\ algorithm running in time $2^{\calO(k^2)}$
(see Proposition~\ref{prop:LD-algo}).
We prove that both of these bounds are optimal.

\begin{restatable}{theorem}{locdomsetsolution}
\label{thm:locating-dom-set-sol-size-lb}
Unless the \ETH\ fails, \LD
does not admit
\begin{itemize}[nolistsep]
\item an algorithm running in time $2^{o(k^2)} \cdot {\rm poly}(n)$,
nor
\item a polynomial-time kernelization algorithm that reduces the solution size
and outputs a kernel with $2^{o(k)}$ vertices.
\end{itemize}
\end{restatable}
To the best of our knowledge, \LD is
the first known problem to admit such an algorithmic lower bound,
with a matching upper bound, when parameterized by the solution size.
The only other problems known to us, admitting similar lower bounds, are for
structural parameterizations like vertex
cover~\cite{DBLP:journals/toct/AgrawalLSZ19,DBLP:journals/corr/abs-2309-02876,FoucaudGK0IST24}
or pathwidth~\cite{DBLP:conf/mfcs/Pilipczuk11,DBLP:journals/iandc/SauS21}.
The second result is also quite rare in the literature.
The only results known to us about \ETH-based conditional lower
bounds on the number of vertices in a kernel when parameterized
by the solution size are for
\textsc{Edge Clique Cover}~\cite{DBLP:journals/siamcomp/CyganPP16} and \textsc{Biclique Cover}~\cite{DBLP:conf/iwpec/ChandranIK16}\footnote{Additionally,
{\sc Point Line Cover} does not admit a kernel with
$\calO(k^{2-\epsilon})$ {\em vertices}, for any $\epsilon >0$, unless
$\NP \subseteq \coNP/poly$~\cite{DBLP:journals/talg/KratschPR16}.}.
Theorem~\ref{thm:locating-dom-set-sol-size-lb} also improves upon a
``no $2^{\calO(k)}n^{\calO(1)}$ algorithm’’ bound from~\cite{BIT20}
(under $\W[2]\neq\FPT$) and a $2^{o(k\log k)}$ \ETH-based lower bound
recently proved in~\cite{DBLP:journals/corr/abs-2011-14849}.

Now, consider the case of \TCPB. As mentioned before,
it is safe to assume that $|\calF| \le 2^{|U|}$ and $|U| \le 2^{k}$.
By Bondy’s celebrated theorem~\cite{bondy1972induced},
which asserts that in any feasible instance of \TCPB, there is always a solution of
size at most $|U| - 1$, we can also assume that $k \le |U| - 1$.
Hence, the brute-force algorithm that enumerates all the
sub-collections of tests of size at most $k$ runs in time
$|\calF|^{\calO(|U|)} = 2^{\calO(|U|^2)} = 2^{2^{\calO(k)}}$.  Our
next result proves that this simple algorithm is again optimal.

\begin{restatable}{theorem}{testcoversolsize}
\label{thm:TC-solsize}
Unless the \ETH\ fails,
\TCPB\ does not admit
\begin{itemize}[nolistsep]
\item an algorithm running in time $2^{2^{o(k)}} \cdot (|U|+|\mathcal F|)^{\calO(1)}$, nor
\item a polynomial-time kernelization algorithm that reduces the solution size and outputs a kernel with $2^{2^{o(k)}}$ vertices.
\end{itemize}
\end{restatable}

This result adds \TCPB\ to the relatively rare list of \NP-complete
problems that admit such double-exponential lower bounds
when parameterized by the solution size, and have a matching algorithm.
The only other examples that we know of are \textsc{Edge Clique Cover}~
\cite{DBLP:journals/siamcomp/CyganPP16},
\textsc{Distinct Vectors Problem}~\cite{DBLP:journals/dmtcs/PilipczukS20},
and \textsc{Telephone Broadcast}~\cite{tale2024}.
For double-exponential algorithmic lower bounds with respect to structural
parameters, please see
\cite{DBLP:journals/talg/FominGLSZ19,
HKL24,
JKL23,
DBLP:journals/toct/KnopPW20,
KLMPS24,
kunnemann_et_al:LIPIcs.ICALP.2023.131,
DBLP:conf/sat/LampisMM18,
DBLP:conf/fsttcs/Lokshtanov0SX21}.
The second result in the theorem is a simple corollary
of the first result.
Assume that the problem admits a kernel with $2^{2^{o(k)}}$ vertices.
Then, the brute-force algorithm enumerating all the possible solutions, works in
time $\binom{2^{2^{o(k)}}}{k} \cdot (|U|+|\mathcal F|)^{\calO(1)}$,
which is $2^{k \cdot 2^{o(k)}} \cdot (|U|+|\mathcal F|)^{\calO(1)}$,
which is $2^{2^{o(k)}} \cdot (|U|+|\mathcal F|)^{\calO(1)}$,
contradicting the first result.
To the best of our knowledge, \TCPB\ is the first problem that
admit a double-exponential kernelization lower bound
for the number of vertices when parameterized by solution size,
or by any natural parameter.

\subparagraph{Related Work.}
\label{subsec:related-work}
\LD\ was introduced by Slater in the
1980s~\cite{DBLP:journals/networks/Slater87,slater1988dominating}.
The problem is \NP-complete~\cite{colbourn1987locating},
even for special graph classes such as planar unit disk graphs~\cite{MS09},
planar bipartite subcubic graphs, chordal bipartite graphs, split graphs, and
co-bipartite graphs~\cite{F15}, interval and permutation graphs of
diameter~2~\cite{FMNPV17intervals2}.
By a straightforward
application of Courcelle’s theorem~\cite{C90}, \LD\ is \FPT\ for
parameter treewidth and even cliquewidth~\cite{CMR00}.
Explicit polynomial-time algorithms were given for
trees~\cite{DBLP:journals/networks/Slater87}, block
graphs~\cite{ABLW20}, series-parallel
graphs~\cite{colbourn1987locating}, and
cographs~\cite{FMNPV17intervals1}.
Regarding the approximation
complexity of \LD,
see~\cite{F15,GKM08,S07}.

In~\cite{CGS21}, structural parameterizations of \LD\ were studied. It
was shown that the problem admits a linear kernel for the parameter
max-leaf number; however, under standard complexity assumptions, no
polynomial kernel exists for the solution size, combined with either
the vertex cover number or the distance to a clique. They also provide a
double-exponential kernel for the parameter distance to the cluster. In~\cite{DBLP:conf/ciac/ChakrabortyFMT25}, the authors of the present paper design an improved parameterized algorithm for \LD\ with respect to the vertex cover number and a linear kernel for the feedback edge set number.

It was shown in~\cite{BIT20} that \LD\ cannot be solved in time
$2^{o(n)}$ on bipartite graphs, nor in time $2^{o(\sqrt{n})}$ on
planar bipartite graphs, assuming the
\ETH. Moreover, they also showed that \LD\ cannot be solved in time
$2^{\calO(k)}n^{\calO(1)}$ on bipartite graphs, unless
$\W[2]=\FPT$. Note that the authors of~\cite{BIT20} have designed a
complex framework with the goal of studying a large class of
identification problems related to \LD and similar problems. In
the full version~\cite{DBLP:journals/corr/abs-2011-14849}
of~\cite{CGS21}, it is shown that \LD\ does neither admit a
$2^{o(k\log k)}n^{\calO(1)}$-time nor an $n^{o(k)}$-time algorithm,
assuming the \ETH. In~\cite{DBLP:conf/ciac/ChakrabortyFMT25},
the authors of the present paper showed that \LD\ does not admit a compression algorithm returning an input with a subquadratic number of bits, under standard assumptions.

\TCPB\ was shown to be \NP-complete by Garey and Johnson~\cite[Problem
SP6]{GJ79}, and it is also hard to approximate within a ratio of
$(1-\epsilon)\ln n$~\cite{BHHHLRS03} (an approximation algorithm with
ratio $1+\ln n$ exists by reduction to \textsc{Set
Cover}~\cite{BDK05}). As any solution has size at least $\log_2(n)$,
the problem admits a trivial kernel of size $2^{2^k}$, and thus
\TCPB\ is \FPT\ parameterized by solution size $k$.
\TCPB was studied within the framework of
above/below guarantee'' parameterizations  in~\cite{BFRS16,CGJMY16,CGJSY12,GMY13} and kernelization in~\cite{BFRS16,CGJMY16,GMY13}. These results have shown an intriguing behavior for \TCPB, with some nontrivial techniques being developed to solve the problem~\cite{BFRS16,CGJSY12}. \TCPB\ is \FPT\ for parameters $n-k$, but $W[1]$-hard for parameters $m-k$ and $k-\log_2(n)$~\cite{CGJSY12}. However, assuming standard assumptions, there is no polynomial kernel for the parameterizations by $k$ and $n-k$~\cite{GMY13}, although there exists a partially
polynomial kernel’’ for parameter $n-k$~\cite{BFRS16} (i.e., one with
$O((n-k)^7)$ elements, but potentially exponentially many tests). When all the tests have a fixed upper bound $r$ on their size, the parameterizations by $k$, $n-k$, and $m-k$ all become FPT with a polynomial kernel~\cite{CGJMY16,GMY13}. In~\cite{DBLP:conf/ciac/ChakrabortyFMT25},
the authors of the present paper showed that \TCPB\ can be solved in time $2^{\calO(|U|\log|U|)}(|U|+|\calF|)^{\calO(1)}$, but does not admit a compression algorithm returning an input with a subquadratic number of bits, under standard assumptions.

The problem \textsc{Discriminating Code}~\cite{CCCCHL08} is very
similar to \TCPB\ (with the distinction that the input is presented as
a bipartite graph, one part representing the elements and the other,
the tests, and that every element has to be covered by some solution
test), and has been shown to be \NP-complete even for planar
instances~\cite{CCHL08}.

\subparagraph{Organization.}
We use standard notations which we specify in Section~\ref{sec:prelims}.
We use the \LD\ problem to demonstrate key technical concepts
regarding our lower bounds and algorithms.
We present an overview of the arguments about \LD in
Sections~\ref{sec:LD-tw} and~\ref{sec:LD-sol-size}, respectively, for
parameters treewidth and solution size.
The arguments regarding \TCPB\ follow the same
line and are presented in Section~\ref{sec:modifications-for-test-cover}.
We conclude with some open problems in
Section~\ref{sec:conclusion}.

\section{Preliminaries}
\label{sec:prelims}

For a positive integer $q$, we denote {the} set $\{1, 2, \dots, q\}$ by $[q]$.
We use $\mathbb{N}$ to denote the collection of all non-negative integers.

\subparagraph*{Graph theory.}
We use standard graph-theoretic notation, and we refer the reader 
to~\cite{Diestel12} for any undefined notation. For an undirected graph $G$, 
sets $V(G)$ and $E(G)$ denote its set of vertices and edges, respectively.
We denote an edge with two endpoints $u, v$ as $uv$.
Unless otherwise specified, we use $n$ to denote the number of vertices in 
the input graph $G$ of the problem under consideration.
Two vertices $u, v$ in $V(G)$ are \emph{adjacent} if there is an edge $uv$ {in 
$G$}. 
The \emph{open neighborhood} of a vertex $v$, denoted by $N_G(v)$, is the 
set of vertices adjacent to $v$.
The \emph{closed neighborhood} of a vertex $v$, denoted by $N_G[v]$, is 
the set $N_G(v) \cup \{v\}$.
We say that a vertex $u$ is a \emph{pendant vertex} if $|N_G(v)| = 1$.
We omit the subscript in the notation for neighborhood if the graph under 
consideration is clear.
For a subset $S$ of $V(G)$, we define $N[S] = \bigcup_{v \in S} N[v]$ and 
$N(S) = N[S] \setminus S$.
For a subset $S$ of $V(G)$, we denote the graph obtained by deleting $S$ 
from $G$ by $G - S$.
We denote the subgraph of $G$ induced on the set $S$ by $G[S]$.

The \emph{vertex integrity} of a graph $G$ is the minimum over $|S|+\max_{D\in cc(G-S)}|D|$, where $S$ is a subset of $V(G)$ and $cc(H)$ denotes the set of connected components of $H$. It is known that the vertex integrity of a graph $G$ is an upper bound for the treedepth of $G$, and thus, also of the pathwidth and the treewidth of $G$. See~\cite{DBLP:journals/tcs/GimaHKKO22} for more details.


\subparagraph*{Locating-Dominating Sets.}
A subset of vertices $S$ in graph $G$ is called its \emph{dominating set}
if $N[S] = V(G)$.
A dominating set $S$ is said to be a \emph{locating-dominating set}
if for any two different vertices $u, v \in V(G) \setminus S$, we have $N(u) 
\cap S \neq N(v) \cap S $.
In this case, we say vertices $u$ and $v$ are
\emph{distinguished} by the set $S$.
We say a vertex $u$ is \emph{located} by set $S$
if for any vertex $v \in V(G) \setminus \{u\}$,
$N(u) \cap S \neq N(v) \cap S$. By extension,
a set $X$ is \emph{located} by $S$ if all vertices
in $X$ are located by $S$.
We note the following simple observation (see also \cite[Lemma 5]{chakraborty2024n2boundlocatingdominatingsetssubcubic}).

\begin{observation}
\label{obs:nbr-of-pendant-vertex-in-sol}
If $S$ is a locating-dominating set of a graph $G$,
then there exists a locating-dominating set $S'$ of $G$
such that $|S'| \le |S|$ and that contains all vertices that are adjacent a pendant vertices (i.e. vertices of degree~1) in $G$.
\end{observation}
\begin{claimproof}
Let $u$ be a pendant vertex which is adjacent to a vertex $v$ of $G$. We now look for a locating-dominating set $S'$ of $G$ such that $|S'| \le |S|$ and contains
the vertex $v$.
As $S$ is a (locating) dominating set,  we have $\{u, v\} \cap S \neq 
\emptyset$. If $v \in S$, then take $S' = S$.
Therefore, let us assume that $u \in S$ and $v \not\in S$.
Define $S' = (S \cup \{v\}) \setminus \{u\}$.
It is easy to see that $S'$ is a dominating set.
If $S'$ is not a locating-dominating set, then there exists $w$, apart from 
$u$,
in the neadjacenthood of $v$ such that both $u$ and $w$ are adjacent to \emph{only} $v$ in $S'$.
As $u$ is a pendant vertex and $v$ its unique neighbor, $w$ is not adjacent to $u$.
Hence, $w$ was not adjacent any vertex in $S' \setminus \{v\} = S 
\setminus \{u\}$.
This, however, contradicts the fact that $S$ is a (locating) dominating set.
Hence, $S'$ is a locating-dominating set and $|S'| = |S|$. Thus, the result follows from repeating this argument for each vertex of $G$ adjacent to a pendant vertex.
\end{claimproof}


\subparagraph*{Parameterized complexity.}
\label{prelim:pc}
An instance of a parameterized problem $\Pi$ {consists} of an input $I$, 
which is an input of the non-parameterized version of the problem, and an 
integer $k$, which is called the \emph{parameter}.
Formally, $\Pi \subseteq \Sigma^* \times \mathbb{N}$.
A parameterized problem $\Pi$ is \emph{fixed-parameter tractable}, or {\FPT}, if 
given an instance $(I,k)$ of $\Pi$, we can correctly decide whether  $(I,k)$ is a 
\yes-instance 
of $\Pi$ in  time $f(k)\cdot |I|^{\calO(1)}$.
Here, {$f: \mathbb{N} \mapsto \mathbb{N}$} is some computable function 
{depending} only on $k$.
A parameterized problem $\Pi$ is said to admit a {\em kernelization} if there exists an algorithm that given an instance $(I, k)$ of $\Pi$, runs in time polynomial in $|I| + k$, and constructs an instance $(I', k')$ of $\Pi$ such that 
\begin{itemize}
	\item[(i)] $(I, k) \in \Pi$ if and only if $(I', k')\in \Pi$, and
	\item[(ii)] $|I'| + k' \leq 
g(k)$ for some computable function $g: \mathbb{N} \mapsto \mathbb{N}$ 
depending only on $k$. 
\end{itemize}

If $g(\cdot)$ is a polynomial function, then $\Pi$ is said to admit a {\em polynomial} kernelization.
For {a} detailed introduction to parameterized complexity and related 
terminologies, we refer the reader to the recent books by Cygan et 
al.~\cite{cygan2015parameterized} and Fomin et 
al.~\cite{fomin2019kernelization}.

\section{\LD\ Parameterized by Treewidth}
\label{sec:LD-tw}

In the first subsection, we present a dynamic programming algorithm 
that solves the problem in $2^{2^{\calO(\tw)}} \cdot n^{\calO(1)}$.
In the second subsection, we prove that this dependency on treewidth
is optimal, upto the multiplicative constant factors, under the \ETH.

\subsection{Upper Bound}
\label{sub-sec:lds-tw-algo}

In this subsection, we present a dynamic programming (DP) based algorithm for the 
\LD\ problem when parameterized by the treewidth of the input graph.
For completeness, we begin with the necessary definitions.

\begin{definition}[Tree-Decomposition]
\label{defn:tree-decomposition}
A \emph{tree-decomposition} of an undirected graph $G = (V, E)$ is a pair
$\calT = (T, \calX = {X_t} \mid {t \in V(T)})$, where $T$ is a tree and $\calX$ is a collection of subsets of $V(G)$, called \emph{bags}, such that:
\begin{enumerate}[nolistsep]
\item For every vertex $u \in V(G)$, there exists $t \in V(T)$ such that $u \in X_t$.
\item For every edge $uv \in E(G)$, there exists $t \in V(T)$ such that $u, v \in X_t$.
\item For every vertex $u \in V(G)$, the set $\{t \in V(T) \mid u \in X_t\}$ induces a connected subtree of $T$.
\end{enumerate}
Given a tree-decomposition $\calT$, its \emph{width} is defined as $\max_{t \in V(T)}( |X_t| - 1)$. 
The \emph{treewidth} of a graph $G$ is the minimum width over all possible tree-decompositions of $G$.
\end{definition}

To facilitate the description of our dynamic programming algorithm, we use 
the notion of a \emph{nice tree-decomposition}~\cite{niceTW}.

\begin{definition}[Nice Tree-Decomposition]
A rooted tree-decomposition $\calT = (T, {X_t}\mid {t \in V(T)})$ is said to be \emph{nice} 
if every node $t \in V(T)$ has at most two children, and is of one of the following types:
\begin{enumerate}[nolistsep]
\item \emph{Root node:} A node $r$ with $X_r = \emptyset$ and no parent.
\item \emph{Leaf node:} A node $t$ with $X_t = \emptyset$ and no children.
\item \emph{Introduce node:} A node $t$ with a unique child $t'$ such that $X_t = X_{t'} \cup \{u\}$ for $u \not\in X_{t'}$.
\item \emph{Forget node:} A node $t$ with a unique child $t'$ such that $X_{t} = X_{t'} \setminus \{u\}$ for $u \in X_{t'}$.
\item \emph{Join node:} A node $t$ with exactly two children $t_1$ and $t_2$ such that $X_t = X_{t_1} = X_{t_2}$.
\end{enumerate}
\end{definition}
For each node $t \in V(T)$, we consider the subtree $T_t$ of $T$ rooted at $t$.
Let $G_t$ denote the subgraph of $G$ induced by the vertices that appear 
in the bags of nodes in $T_t$.

Without loss of generality, we assume that a nice tree-decomposition of width $\calO(\tw)$ is provided. 
If not, one can be constructed in time $2^{\calO(\tw)} n$~\cite{niceTW,K21}.

We now state the main result of this subsection.

\begin{theorem}
\label{thm:LD-tw-algo}
\LD\ admits an algorithm with running time
$2^{2^{\calO(\tw)}} \cdot n$, where $\tw$ is the treewidth
and $n$ is the order of the input graph.
\end{theorem}

Consider a locating-dominating set \( S \subseteq V(G) \), and let \( S_t \subseteq V(G_t) \) denote the 
\emph{partial solution} induced by restricting \( S \) to the subgraph \( G_t \). Formally, we define
\(S_t := S \cap V(G_t).\)

We now informally describe the information needed for a DP-state for a node $t$ of a tree-decomposition, in order to convey to the reader the main ideas needed to understand the algorithms. The formal proof follows.

Observe that every vertex in \( V(G_t) \setminus X_t \) is uniquely located by its neighborhood 
within \( S_t \), since such vertices will not appear 
in any future bag of the decomposition.

\begin{figure}[t]
\centering
\includegraphics[scale=0.6]{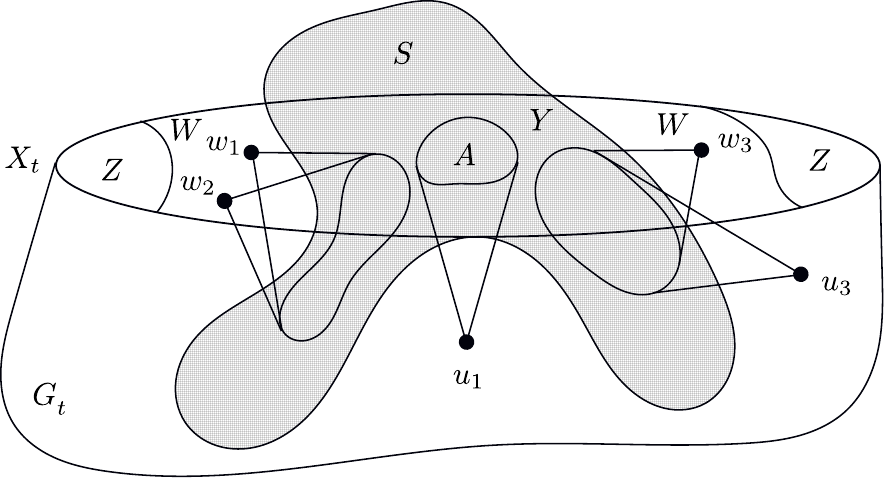}
\caption{The figure illustrates the definitions used in the dynamic programming algorithm.
Let \(S\) be a locating-dominating set of \(G\).
Its intersections with \(V(G_t)\) and \(X_t\) are denoted by \(S_t\) and \(Y\), respectively.
The region containing the vertices \(w_1, w_2,\) and \(w_3\) represents subset \(W \subseteq X_t\).
The remaining vertices of \(X_t\) form the set \(Z\).
Consider a subset \(A \subseteq Y\).
Observe that the vertex \(u_1\) is adjacent only to the vertices in \(A\) within \(S\) (and hence within \(S_t\)).
Therefore, the set \(A\) is not suitable for locating purposes.
\(\mathcal{Y}\) denote the collection of all such subsets.
Note that the vertices \(w_1\) and \(w_2\) have identical neighborhoods in \(S_t\).
Similarly, the vertices \(w_3\) and \(u_3\) have identical neighborhoods in \(S_t\).
Since both \(w_1\) and \(w_2\) belong to \(X_t\), we add the pair \((w_1, w_2)\) to \(\mathbb{W}\).
However, as \(u_3\) is \emph{forgotten} (that is, \(u_3 \notin X_t\)), we add the pair \((w_3, +)\) to \(\mathbb{W}\).
By the definition of \(W_0\), all three vertices \(w_1, w_2\) and \(w_3\) are in it.
 \label{fig:treewidth-DP}}
\end{figure}

Let us now examine the behavior of vertices in the current bag \( X_t \). 
See Figure~\ref{fig:treewidth-DP} for an illustration.
Define:
\begin{itemize}[nolistsep]
  \item \( Y := S_t \cap X_t \), i.e., the subset of the current bag that is included in the partial solution.
  \item \( W := (N(S_t) \cap X_t) \setminus Y \), i.e., the set of vertices in \( X_t \setminus Y \) that are 
  dominated by vertices in \( S_t \) (but may not be located).
  {Also, define \(W_0 = \{ w \in W \mid N_{G_t}(w) \cap S_t \subseteq Y\}\)}.
  {Note that} every vertex in \(W \setminus W_0\) is adjacent to some vertex in \(S_t \setminus Y\).
\item {\(Z := X_t \setminus (Y \cup W)\)}.
\end{itemize}

Note that the vertices in ${Z}$ are not dominated by $S_t$.

Since \( S \) is a locating-dominating set, it must hold that no two vertices outside of \( S \) 
have identical neighborhoods within \( S \). In particular, if there exists a vertex \( u \in V(G_t) \) 
such that \( N(u) {\cap S} \subseteq S_t \), 
then we must ensure that no vertex \( v \in V(G) \setminus V(G_t) \) 
satisfies \( {N(v)\cap S = N(u)\cap S} \). Otherwise, \( u \) and \( v \) would be indistinguishable in \( G \), 
violating  the location constraint.
We remark that for the above case to happen, it must be that 
${N(u) \cap S \subseteq (S_t \cap X_t)}$.
To facilitate the extension of the partial solution \( S_t \) to a full solution \( S \subseteq V(G) \), 
we track which subsets of \( Y \) are used to identify vertices. 
Specifically, we consider subsets \( A \subseteq Y \) such that there exists 
\(u \in V(G_t) \setminus S_t \text{ with } N_{G_t}(u) \cap S_t = A.\)
If \( u \in V(G_t) \setminus X_t \), then \( A \) is \emph{not usable for location} in future extensions, 
since no vertex added in later bags can change the neighborhood of $u$ in the solution.
Suppose that \( \mathcal{Y} \) denotes the collection of subsets of \( Y \) that are not usable for future location.

For the reason hinted above, we need to distinguish between the vertices in \(W\) that are only adjacent 
to a subset of \(Y\) and those that are are adjacent to some vertices in \(S_t \setminus Y\) also. 
{This is why we defined $W_0$ above.}

Finally, we define a set \( \mathbb{W} \) to keep track of pairs of vertices that 
currently have identical neighborhoods in \( S_t \) and must be distinguished in 
future extensions. 
Each element of \( \mathbb{W}\) is of one of the 
following forms:
\begin{itemize}[nolistsep]
  \item A pair \( (w_1, w_2) \) {in $X_t$} and more specifically in \(W\), 
  indicating that \( w_1 \) and \( w_2\) currently have identical 
  neighborhoods within \( S_t \), and some vertex in \( S \setminus S_t \) must be adjacent 
  to exactly one of them to distinguish them.

  \item A pair \( (w_1, +) \), indicating that \( w_1 \in W \) and some vertex 
  \( u \in V(G_t) \setminus X_t \) share the same neighborhood in \( S_t \), and 
  hence a vertex in \( S \setminus S_t \) must be adjacent to \( w_1 \) to ensure distinguishability.
\end{itemize}
Note that for any \( (w_1, w_2) \in \mathbb{W}\), vertex \(w_1\) is in  \(W_0\) (respectively,
	\(W \setminus W_0\)) if and only if \(w_2\) is in \(W_0\) (respectively,
	\(W \setminus W_0\)).

\smallskip
For every node \( t \in V(T) \), we define a DP-state using the notion of a \emph{valid tuple},
which is a combination of the sets described above.
\begin{definition}[Valid Tuple]
\label{def:valid-tuple}
Let \( t \in V(T) \) be a node in the tree-decomposition. 
A tuple \(\langle Y, W, W_0, \calY, \mathbb{W} \rangle\) is said to be a 
\emph{valid tuple} at $t$ if the following conditions hold:
\begin{itemize}[nolistsep]
  \item \( Y, W \subseteq X_t \) are disjoint sets such that \( N(Y) \cap X_t \subseteq W \),  \(W_0 \subseteq W\),
  \item \( \calY \) is a family of subsets of \( Y \),
  \item \( \mathbb{W} \) {is a collection of pairs which can be of two types:}
  pairs \( (w_1, w_2) \) with \( w_1, w_2 \in W \), and
  pairs of the form \( (w_1, +) \), with \( w_1 \in W \).
\end{itemize}
\end{definition}

{We now define how a valid tuple for node $t$ can correspond to a potential subset $S_t$ of a locating-dominating $S$ of $G$.}

\begin{definition}[Candidate Solution]
\label{def:candidate-sol}
Consider a valid tuple \( \tau = \langle Y, W, W_0, \calY, \mathbb{W} \rangle\) at node $t$.
We say that set \( S_t \subseteq V(G_t) \) is a \emph{candidate solution} at $t$ with respect to \(\tau\) 
if it satisfies the following properties:
\begin{enumerate}[nolistsep]
  \item \( S_t \) {locates (in $G_t$)} all vertices in \( V(G_t) \setminus X_t \), i.e., 
  for any distinct vertices \( u, v \in V(G_t) \setminus (S_t \cup X_t) \), the sets \( N_{G_t}(u) \cap S_t \) and 
  \( N_{G_t}(v) \cap S_t \) are nonempty and distinct.

  \item The intersection of \( S_t \) with the current bag is exactly \( Y \), and the vertices in 
  \( X_t \setminus Y \) that are dominated by \( S_t \) are given by \( W \), i.e., 
  \(Y = S_t \cap X_t,\)  and \( W =  (N(S_t) \cap X_t) \setminus Y \).
  Moreover, no vertex in $W_0$ is adjacent to a vertex in $S_t\setminus Y$ and
  every vertex \(w \in W \setminus W_0\) is adjacent to some vertex in \(S_t \setminus Y\).
  
  \item A subset $A$ belongs to \( \calY \) if and only if there exists a vertex \( u \in V(G_t) \setminus (S_t \cup X_t) \) 
  such that \( N_{G_t}(u) \cap S = A \). {(By the first item, if $u$ exists it is unique.)} 


  \item 
  \begin{itemize}
	\item  A pair \( (w_1, w_2) \in \mathbb{W} \) if and only if \( w_1, w_2 \in W \) and
  \(N_{G_t}(w_1) \cap S_t = N_{G_t}(w_2) \cap S_t\).
  
	\item A pair \( (w_1, +) \in \mathbb{W} \) if and only if {$w_1\in W$ and} there exists a vertex 
 \( u \in V(G_t) \setminus (S_t \cup X_t) \) such that
  \(N_{G_t}(w_1) \cap S_t = N_{G_t}(u) \cap S_t\). {(By the first item, if $u$ exists it is unique.)}
  \end{itemize}
\end{enumerate}
\end{definition}

For a valid tuple \( \tau = \langle Y, W, W_0, \calY, \mathbb{W} \rangle\) at node $t$, 
we define \(\dptw[t, \tau]\) as the minimum cardinality of a candidate solution at $t$.
If no candidate solution exists, we define
\( \dptw[t, \tau] := \infty \).

\smallskip
We now proceed to describe a recursive algorithm used to update 
entries in the dynamic programming table 
for each node type in a nice tree-decomposition.

\subparagraph*{Leaf node.}
If $t$ is a leaf node, then $X_t$ is an empty set
and the following claim is trivial.
\begin{lemma}
\label{lemma:leaf-node}
If $t$ is a leaf node, then
$\dptw[t, \tau] = 0$ for
$\tau = \langle\emptyset, \emptyset, \emptyset, \emptyset, \emptyset \rangle$
and $\dptw[t, \tau] = \infty$ for any 
other $\tau$.
\end{lemma}

\subparagraph*{Introduce Node.}
Let \( t \in V(T) \) be an {introduce node} with unique 
child \( t' \), 
such that \(X_t = X_{t'} \cup \{x\}\),
for some vertex \( x \in V(G_t) \setminus V(G_{t'}) \). 
Let \(\tau = \langle Y, W, W_0, \calY, \mathbb{W} \rangle\)
be a valid tuple at node \( t \). 
We present the update rule for \( \dptw[t, \tau] \) 
based on where $x$ is in the set $X_t$.
We need the following two definitions for the lemma.

\begin{definition}[Introduce-\( Y \)-compatible]
\label{def:introduce-Y-compatible}
{Consider the case when $x$ is in $Y$.}
Consider a tuple \( \tau' = \langle Y', W', W'_0, \calY', \mathbb{W}' \rangle \) which is 
valid at $t'$.
We say $\tau'$ is \emph{introduce-\( Y \)-compatible} 
with $\tau$ if it satisfies the following conditions:
\begin{enumerate}[nolistsep]
\item \( Y' = Y \setminus \{x\} \), \(W =  W' \cup (N_{G_t}(x) \setminus {N_{G_t}[Y]}) \), \(W_0 = W'_0 \cup (N_{G_t}(x) \setminus {N_{G_t}[Y]})\),
\item \( \calY' = \calY \), and 
\item \( \mathbb{W} \subseteq \mathbb{W}' \) such that 
\begin{itemize}[nolistsep]
\item for each element of the type
\( (w_1, w_2) \in \mathbb{W} \setminus \mathbb{W}' \), 
exactly one of \( w_1 \) and \( w_2 \) is adjacent to \( x \); and 
\item for each element of the type \( (w_1, +) \in \mathbb{W} \setminus \mathbb{W}' \), 
\( w_1\) is adjacent to \( x \) in \( G_t \).
\end{itemize}
\end{enumerate}
Let \( I^Y_c(\tau) \) be the collection of all valid tuples 
at $t'$ that are {introduce-\( Y \)-compatible} 
with \( \tau \).
\end{definition}

{Consider the case when $x$ is in $W$.}
Define \( A := N_{G_t}(x) \cap Y \) as the set of 
neighbors of \( x \) that are included in the partial solution. 
If \( A = \emptyset \), then \( x \) is not adjacent 
to any vertex in the partial solution, 
violating the definition of set $W$. 
In this case, set \(\dptw[t, \tau] := \infty\).

\begin{definition}[Introduce-\( W \)-compatible]
\label{def:introduce-W-compatible}
Consider the case when $x$ is in $W$ and suppose \( A := N_{G_t}(x) \cap Y\neq \emptyset \). 
Consider a valid tuple  \( \tau' = \langle Y', W', W'_0, \calY', \mathbb{W}' \rangle \) at \( t' \). 
We say that \( \tau' \) is \emph{introduce-\( W \)-compatible} with \( \tau \) if:
\begin{enumerate}[nolistsep]
\item \( Y' = Y \),  {$W=W'\cup\{x\}$, $W_0=W'_0\cup\{x\}$,} 
\item \( \calY' = \calY \), and 
\item \(\mathbb{W}\) is obtained from \(\mathbb{W}'\) by the following procedure.
  \begin{enumerate}[nolistsep]
  \item If \( A \in \calY'\), i.e., {$A$ models some $u \in V(G_{t'}) \setminus X_{t'}$
  such that $N(u) \cap S_{t'} = A$ in a potential candidate solution $S_{t'}$ at $t'$}, then add \((x, +)\) to \(\mathbb{W}'\). 
 \item For every vertex $w_1$ in $W'_0$ such that 
{ \( A = N_{G_{t'}}(w_1) \cap Y \),}
 add \((x, w_1)\) to \(\mathbb{W}'\).
  \end{enumerate}
\end{enumerate}
Let \( I^W_c(\tau) \) denote the set of all introduce-\( W \)-compatible tuples at \( t' \).
\end{definition}

\begin{lemma}
\label{lemma:introduce-node}
Let \( t \) be an introduce node with child \( t' \), where \( X_t = X_{t'} \cup \{x\} \), and let 
\( \tau = \langle Y, W, W_0, \calY, \mathbb{W} \rangle \) be a valid tuple at \( t \). Then,
\[
\dptw[t, \tau] =
\begin{cases}
1 + \min\limits_{\tau' \in I^Y_c(\tau)} \dptw[t', \tau'] & \text{if } x \in Y, \\
\min\limits_{\tau' \in I^W_c(\tau)} \dptw[t', \tau'] & \text{if } x \in W, \text{ and } N_{G_t}(x) \cap Y \neq \emptyset \\
\dptw[t', \tau] & \text{if }  {x \in Z}\\
{\infty} & \text{otherwise}.
\end{cases}
\]
\end{lemma}
\begin{proof}
As discussed earlier, we consider three mutually disjoint and exhaustive cases based 
on the membership of $x$ in the sets specified by $\tau$: whether $x \in Y$, $x \in W$, or 
$x \in {Z}$.

\smallskip

\textbf{Case 1: $x$ is in  $Y$.} 
Assume that the minimum value of \(\min\limits_{\tau' \in I^Y_c(\tau)} \dptw[t', \tau']\) 
is finite and attained at $\tau'$ and is witnessed by the set $S_{t'}$. 
Define $S_t := S_{t'} \cup \{x\}$. 
We show that $S_t$ is a valid candidate solution at $t$ by verifying that it satisfies 
the conditions in Definition~\ref{def:candidate-sol}.
We refer the readers to Definition~\ref{def:introduce-Y-compatible} 
for the relationship between the sets specified by $\tau$ and $\tau'$.

\begin{enumerate}[nolistsep]
    \item As \(S_{t'}\) {locates} all vertices in \(V(G_{t'}) \setminus X_{t'}\),
     set \(S_t \), which is a superset of \(S_{t'}\), is a locating-dominating set for all vertices in 
     \(V(G_t) \setminus X_t\), which is {equal to} \(V(G_{t'}) \setminus X_{t'}\).
    
    \item As \(S_t = S_{t'} \cup \{x\}\) and \(Y = Y' \cup \{x\}\),
    the intersection of \(S_t\) with $X_t$ is exactly \(Y\). 
    Moreover, the vertices in \(X_t \setminus Y\) that are dominated by \(S_t\) form the set \(W\), as
    \(W =  W' \cup (N_{G_t}(x) \setminus N_{G_t}[Y]) \).
    Also, vertices in \(W'\) whose neigbhors in the partial soslution $S_{t'}$ are all in \(Y'\) remain the same,
    apart from those that are adjacent to only \(x\).
    Hence, \(W_0 = W'_0 \cup (N_{G_t}(x) \setminus {N_{G_t}[Y]}) \).
    
    \item Since $x$ is introduced in this bag, $x$ is not adjacent to any vertex in \(V(G_t) \setminus X_t\). 
    Thus, the collection of unusable subsets $\mathcal{Y}$ at $t$ is identical to the collection 
    $\mathcal{Y}'$ at $t'$.
    
    \item  {Now, consider the set \(W\). As only change in partial dominating set is
    the addition of vertex \(x\), the only new vertices that are added to \(W\) are from \(N_{G_t}(x)\),
    and hence we have \({\mathbb{W} \subseteq \mathbb{W}'}\).}
    Furthermore, for any pair \((w_1, w_2)\) or \((w_1, +)\) which is in \(\mathbb{W}'\) but not in \(\mathbb{W}\), 
     the vertex $x$ must locate one of these vertices:
    \begin{itemize}[nolistsep]
        \item In the case \((w_1, w_2)\), $x$ is adjacent to exactly one of $w_1$ or $w_2$.
        \item In the case \((w_1, +)\), $x$ is adjacent to $w_1$.
    \end{itemize}
\end{enumerate}
This confirms that \(S_t = S_{t'} \cup \{x\}\) is a valid candidate solution for $\dptw[t, \tau]$.
As \(x\) is introduced at this node, \(x \not\in S_{t'}\) and hence \(|S_t| = |S_{t'}| + 1\).
Thus we have:
\(\dptw[t, \tau] \le {|S_t| = 1+|S_{t'}| = } 1 + \dptw[t', \tau']\).

\smallskip

Conversely, suppose that $S_t$ is an optimal candidate solution at $t$ corresponding to $\tau$. 
We show that \(S_{t'} := S_t \setminus \{x\}\) is a valid candidate solution at $t'$ for $\tau'$.

\begin{enumerate}[nolistsep]
    \item Since $x$ is adjacent only to vertices in $X_t$ within $G_t$, and $S_t$ is a locating-dominating set for \(V(G_t) \setminus X_t\), it follows that \(S_{t'}\) is a locating-dominating set for \(V(G_{t'}) \setminus X_{t'}\).
    
    \item By construction of \(S_{t'}\), we have \(S_{t'} \cap X_{t'} = Y'\). Additionally, 
    since \(W =  W' \cup (N_{G_t}(x) \setminus {N_{G_t}[Y]}) \), 
    $W'$ captures precisely those vertices in \(X_{t'}\) that are dominated by some vertex in \(S_{t'}\).
    Using the same argument, \(W'_0\) precisely contains the vertices in \(W'\) whose neighborhood
    in the partial solution is  in \(Y'\).
    
    \item A subset \(A \in \mathcal{Y}\) if and only if there exists a unique vertex 
    \(u \in V(G_t) \setminus (S_t \cup X_t)\) such that \(N_{G_t}(u) \cap S_t = A\).
     Since $x$ is not adjacent to any vertex in \(V(G_{t'}) \setminus X_{t'}\), the collection of unusable 
     sets remains unchanged, i.e., \(\mathcal{Y}' = \mathcal{Y}\). 
     
    \item Note that $x$ may locate certain vertex pairs in $W'$ that were not located by $S_{t'}$. Hence, we have \(\mathbb{W} \subseteq \mathbb{W}'\). More precisely:
    \begin{itemize}[nolistsep]
        \item For \((w_1, w_2) \in \mathbb{W}' \setminus \mathbb{W}\), this occurs only if exactly one of \(w_1\) or \(w_2\) is adjacent to $x$.
        \item For \((w_1, +) \in \mathbb{W}' \setminus \mathbb{W}\), this occurs only if \(w_1\) is adjacent to $x$.
    \end{itemize}
\end{enumerate}
Thus, \(S_{t'}\) is a valid candidate solution at $t'$.
As \(S_t\) is a valid candidate at $t$, and $x$ is in $Y$,
this implies that \(x\) is in \(S_t\), and
hence, \(|S_{t'}| = |S_t| - 1\). {Thus, $\dptw[t', \tau']\leq |S_{t'}| = |S_t|-1 = \dptw[t, \tau]-1$.}
We conclude:
\(\dptw[t, \tau] \ge 1 + \dptw[t', \tau']\).

The above arguments establish the correctness of the recursive computation at introduce nodes 
when the introduced vertex $x$ is included in the candidate solution.

\smallskip

\textbf{Case 2: \(x\) is in \(W\).} 
Assume that the minimum value of \(\min\limits_{\tau' \in I^W_c(\tau)} \dptw[t', \tau']\) 
is finite and attained at \(\tau'\) and witnessed by a set \(S_{t'}\). 
We argue that \(S_t := S_{t'}\) is also a valid candidate solution 
at \(t\) by verifying that it satisfies the conditions in Definition~\ref{def:candidate-sol}.
We refer the reader to Definition~\ref{def:introduce-W-compatible} 
for the correspondence between the sets specified by \(\tau\) and \(\tau'\).
Moreover, as noted in the discussion preceding Definition~\ref{def:introduce-W-compatible},
it suffices to consider the case when \(A := N_{G_t}(x) \cap Y \neq \emptyset\).

\begin{enumerate}[nolistsep]
  \item Since \(S_{t'}\) {locates} all vertices in \( V(G_{t'}) \setminus X_{t'} \), 
  the same holds for \(S_t\) with respect to \( V(G_t) \setminus X_t \).

  \item As no new vertex is included in the candidate solution, the intersection 
  \( S_t \cap X_t \) equals \( Y = Y' \). 
  Furthermore, since the newly introduced vertex belongs to \(W\),
  we have \( W' = W \setminus \{x\} \).
  Note that \(x\) can only be adjacent to vertices in $S_t$ 
  that are in $X_t$, and hence {$W_0=W'_0\cup\{x\}$.} 

  \item Again, because the candidate solution remains unchanged, the collection of unusable subsets
  \(\calY\) at \(t\) is identical to \(\calY'\) at \(t'\).

  \item As no new vertex is added to the candidate solution, no pairs from \(\mathbb{W}'\) are located, and thus \(\mathbb{W}' \subseteq \mathbb{W}\). 
  Any pair in \(\mathbb{W} \setminus \mathbb{W}'\) necessarily involves \(x\).
  The other vertex in the pair depends on the structure of the set \(A\). We consider two subcases:
  \begin{enumerate}[nolistsep]
    \item If \((x, +) \in \mathbb{W} \setminus \mathbb{W}'\), then \(A \in \calY = \calY'\); 
    that is, there exists a vertex \(u \in V(G_{t'}) \setminus X_{t'}\) such that 
    \(N(u) \cap S_{t'} = A\).

    \item If \((x, w_1) \in \mathbb{W} \setminus \mathbb{W}'\), then \(w_1 \in W'\) and 
    \(A = N_{G_{t'}}(w_1) \cap S_{t'}\).
  \end{enumerate}
\end{enumerate}

This verifies that \(S_t = S_{t'}\) is a valid candidate solution for \(\dptw[t, \tau]\),
and hence we obtain: \(\dptw[t, \tau] \le \dptw[t', \tau']\).

\smallskip

Conversely, suppose that \(S_t\) is an optimal candidate solution at node \(t\) corresponding to tuple \(\tau\). 
We show that \(S_{t'} := S_t\) is a valid candidate solution at \(t'\) for tuple \(\tau'\).

\begin{enumerate}[nolistsep]
  \item Since \(x\) is introduced in this bag, 
  \( S_{t'} \) {locates} all vertices in \( V(G_{t'}) \setminus X_{t'} \).

  \item The partial solution remains unchanged, so \(Y = Y'\).
  By definition, \( W =  W' \cup \{x\} \),
  and hence \(W'\) consists of all vertices in \(V(G_t)\) that are dominated
  by \(S_{t'}\).
  Also, as \(x\) can only be adjacent to vertices in $S_t$
  that are in the bag, we have {$W_0=W'_0\cup\{x\}$.} 
  
  \item Again, since the partial solution remains unchanged, the set of unusable 
  subsets of \(Y\) is preserved.  
  
  \item As \(W = W' \cup \{x\}\), we again have \(\mathbb{W}' \subseteq \mathbb{W}\).
  Every pair in  \(\mathbb{W} \setminus \mathbb{W}'\) must involve \(x\),
  and the other component of the pair depends on how the neighborhood of \(x\) intersects \(Y\),
  as detailed in Definition~\ref{def:introduce-W-compatible}.
\end{enumerate}

This confirms that \(S_{t'} = S_t\) is a valid candidate solution at \(t'\), and thus 
\(\dptw[t, \tau] \ge \dptw[t', \tau']\).

Combining both directions, we establish the correctness of the recursive update at 
introduce nodes when the newly introduced vertex \(x\) is not included in the partial solution
but is dominated by some vertex already in the partial solution.

\smallskip
\textbf{Case 3: \(x\) is  in \({Z}\):} 
If \(N_{G_t}[x] \cap Y \neq \emptyset\), then \(x\) is either in the candidate solution 
or adjacent to a vertex in the candidate solution, contradicting the assumption that \(x\) lies 
in \({Z}\). 
In this case, we set \(\dptw[t, \tau] := \infty\).
Otherwise, the new vertex \(x\) imposes no additional constraint on the partial solution.  
Therefore, \(\tau\) is also a valid tuple at \(t'\), and the DP value remains unchanged:
\(\dptw[t, \tau] = \dptw[t', \tau]\).

This completes the proof of the lemma.
\end{proof}

\subparagraph*{Forget Node.}
Let \( t \in V(T) \) be a {forget node} with unique child \( t' \), 
such that \(X_t = X_{t'} \setminus \{x\}\),
for some vertex \( x \in X_{t'} \).
Let \(\tau = \langle Y, W, W_0, \calY, \mathbb{W} \rangle\)
be a valid tuple at node \( t \). 
The following claim presents the update rule 
for \( \dptw[t, \tau] \) based on which set $x$ belongs to in $\tau'$. 
We follow the same template as before, we present definitions
of sets $F^Y_c$ and $F^W_c$ which are \emph{compatible} with $\tau$
before stating the lemma.

\begin{definition}[Forget-\({Y} \)-compatible]
\label{def:forget-Y-compatible}
Consider tuple  
\( \tau' = \langle Y', W', W'_0, \calY', \mathbb{W}' \rangle \) which is valid at \( t' \). 
{Suppose $x$ is in $Y'$.}
We say $\tau'$ is \emph{forget-\( Y\)-compatible} with
\( \tau \) if it satisfies the following conditions:
\begin{enumerate}[nolistsep]
\item \( Y = Y' \setminus \{x\} \), \( W = W' \), \( W_0 = W'_0 \setminus \{w \in W' \mid x \in N_{G_{t'}}(w) \}\),
\item \( \calY = \calY' \setminus \{A \in \calY \mid x \in A\} \), and 
\item \( \mathbb{W} = \mathbb{W}' \).
\end{enumerate}
Let \( F^Y_c(\tau) \) denote the collection of all valid
tuples at $t'$ that are forget-\(Y'\)-compatible with $\tau$.
\end{definition}

Consider the case when $x$ is in $W'$. 

\begin{definition}[Forget-\({W} \)-compatible]
\label{def:forget-W-compatible}
Consider tuple  \( \tau' = \langle Y', W', W'_0, \calY', \mathbb{W}' \rangle \) 
{which is valid} at \( t' \).
{Suppose \(x\) is in \(W'\) and} $(x, +) \not\in \mathbb{W}'$. 
We say that \( \tau' \) is \emph{forget-\( W \)-compatible} with \( \tau \) if:
\begin{enumerate}[nolistsep]
\item \( Y = Y' \),  \( W = W' \setminus \{x\} \), and \( W_0 = W'_0 \setminus \{x\} \),
\item If 
\(x \in W'_0\),  then \(\calY' = \calY \cup \{A\}\) where \(A := N_{G_{t'}}(x) \cap Y'\) 
{because $(x, +) \notin \mathbb{W}'$}; otherwise \( \calY' = \calY \), and
\item $\mathbb{W}$ is obtained from $\mathbb{W}'$ 
by replacing every tuple of the form $(w_1, x)$, for some \(w_1\) in \(W\),
by \((w_1, +)\).
\end{enumerate}
Let \( F^W_c(\tau) \) denote the collection of all valid
tuples at $t'$ that are forget-\(W\)-compatible with $\tau$.
\end{definition}

\begin{lemma}
\label{lemma:forget-node}
Let \( t \) be a forget node with child \( t' \), where \( X_t = X_{t'} \setminus \{x\} \), and 
let \( \tau = \langle Y, W, W_0, \calY, \mathbb{W} \rangle \) be a valid tuple at \( t \). Then,
\[
\dptw[t, \tau] =
\begin{cases}
\min\limits_{\tau' \in F^Y_c(\tau)} \dptw[t', \tau'] & \text{if } x \in Y', \\
\min\limits_{\tau' \in F^W_c(\tau)} \dptw[t', \tau'] & \text{if } x \in W', \text{ and } N_{G_{t'}}(x) \cap Y' \neq \emptyset; 
\text{ and }
(x, +) \not\in \mathbb{W}'\\
\infty & \text{if } {x \in Z'}\\
\infty & otherwise.
\end{cases}
\]
\end{lemma}
\begin{proof}
We consider the three mutually disjoint and 
exhaustive cases depending on whether 
$x$ is in $Y'$, $W'$, or $X_{t'} \setminus (Y' \cup W')$ 
in the sets specified by $\tau'$.

\smallskip

\textbf{Case 1: $x$ is in  $Y'$.} 
Suppose that the minimum value of \(\min\limits_{\tau' \in I^Y_c(\tau)} \dptw[t', \tau']\) 
is finite and attained at \(\tau'\), and let \(S_{t'}\) be the corresponding witnessing set. 
Define \(S_t := S_{t'}\). 
We now verify that \(S_t\) constitutes a valid candidate solution at node \(t\) by checking that it satisfies 
the conditions in Definition~\ref{def:candidate-sol}.
The relationship between the sets corresponding to \(\tau\) and \(\tau'\) is specified by 
Definition~\ref{def:forget-Y-compatible}.

\begin{enumerate}[nolistsep]
    \item Since \(S_{t'}\) locates all vertices in \(V(G_{t'}) \setminus X_{t'}\), and
    \((V(G_{t'}) \setminus X_{t'}) \cup \{x\} = V(G_{t}) \setminus X_{t}\), 
    it follows that \(S_t\) locates all vertices in \(V(G_{t}) \setminus X_{t}\). 
   
    \item As \(X_t = X_{t'} \setminus \{x\}\) and \(Y = Y' \setminus \{x\}\),
    the intersection of \(S_t\) with \(X_t\) is precisely \(Y\). 
    Furthermore, since the partial solution remains unchanged, the set of vertices in \(X_t \setminus Y\) that are dominated by \(S_t\) is denoted by \(W\).
    Every vertex in \(W'\) that is adjacent to \(x\) is now adjacent to some vertex in the
    partial solution which is not in the bag, and hence 
    \( W_0 = W'_0 \setminus \{w \in W' \mid x \in N_{G_{t'}}(w) \}\).
    
    \item Since \(x\) is the forgotten vertex, the collection of unusable subsets \(\mathcal{Y}\) at 
    node \(t\) is derived from \(\mathcal{Y}'\) by removing all subsets that contain \(x\).
    
    \item As the partial solution remains unchanged, the set of  vertex pairs that are not located also remains identical.
\end{enumerate}

Consequently, \(S_t = S_{t'}\) is a valid candidate solution for \(\dptw[t, \tau]\), and we obtain
\(\dptw[t, \tau] \le \dptw[t', \tau'].
\)

\smallskip

Conversely, suppose that \(S_t\) is an optimal candidate solution at node \(t\) corresponding to \(\tau\), and that \(x \in S_t\). 
We demonstrate that the set \(S_{t'} := S_t\) constitutes a valid candidate solution at node \(t'\) for \(\tau'\).

\begin{enumerate}[nolistsep]
    \item Since \(S_t\) locates all vertices in \(V(G_t) \setminus X_t\), and
    \((V(G_{t'}) \setminus X_{t'}) \cup \{x\} = V(G_t) \setminus X_t\),
    it follows that \(S_{t'}\) locates all vertices in \(V(G_{t'}) \setminus X_{t'}\).
    
    \item As \(X_t = X_{t'} \setminus \{x\}\) and \(Y = Y' \setminus \{x\}\),
    the intersection of \(S_{t'}\) with \(X_{t'}\) is precisely \(Y'\). 
    Additionally, as the partial solution remains unchanged, the set of vertices in \(X_{t'} \setminus Y'\) dominated by \(S_{t'}\) is denoted by \(W'\).
    	Any vertex which is in {\(W_0'\)} but not in {\(W_0\)} must be adjacent to \(x\)
    	and hence \( W_0 = W'_0 \setminus \{w \in W' \mid x \in N_{G_{t'}}(w) \}\).
         
    \item Again, since the partial solutions are identical, the collection of unusable subsets at \(t'\) differs from that at \(t\) only by the exclusion of those subsets that include \(x\).
        
    \item Likewise, the collection of vertex pairs that must be located via the extension of the partial solution remains the same.
\end{enumerate}
Hence, \(S_{t'}\) is a valid candidate solution at node \(t'\), which implies
\(
\dptw[t, \tau] \ge \dptw[t', \tau'].
\)

Combining both directions of the argument, we conclude that the recursive computation at 
forget nodes correctly preserves the minimum when the forgotten vertex \(x\) is included in the partial solution.

\smallskip

\textbf{Case 2: $x$ is in  $W'$.} 
Suppose that the minimum value of 
\(\min\limits_{\tau' \in I^Y_c(\tau)} \dptw[t', \tau']\) 
is finite and attained at some \(\tau'\), and let \(S_{t'}\) be the corresponding witnessing set. 
Define \(S_t := S_{t'}\). 
We now verify that \(S_t\) constitutes a valid candidate solution at node \(t\) by checking that 
it satisfies the conditions specified in Definition~\ref{def:candidate-sol}. 
The relationship between the sets defined by \(\tau\) and \(\tau'\) is specified by Definition~\ref{def:forget-W-compatible}. 

Note that if $(x, +) \in \mathbb{W}'$, there are two vertices in $G_t$, one of them being $x$, that have identical neighborhoods in $S_{t'}=S_t$. 
As no new vertex can be adjacent to either of these vertices, these two will have same neighbhorhood in any extension of $S_t$ and thus, \(\dptw[t, \tau] := \infty\). Hence, we may assume next that $(x, +) \notin \mathbb{W}'$.

\begin{enumerate}[nolistsep]
    \item 
    As \(V(G_{t'}) \setminus X_{t'}\) and \(V(G_t) \setminus X_t\) differ only by the vertex \(x\), 
    to prove that \(S_t\) locates \(V(G_t) \setminus X_t\), it is suffices to prove that it locates \(x\).
    This follows from the fact that \((x, +) \not\in \mathbb{W}'\), i.e., no vertex \(u \in V(G_{t'})\) satisfies \(N_{G_t}(u) \cap S_t = N_{G_t}(x) \cap S_t\).

    \item Since the partial solution remains unchanged, we have \(S_t \cap X_t = Y = Y'\). Moreover, because the forgotten vertex belongs to \(W\), we have \(W = W' \setminus \{x\}\) and \( W_0 = W'_0 \setminus \{x\} \).

    \item If 
    \(x \in W_0\), then \(x\) is adjacent only to vertices in \(A = N_{G_t}(x) \cap S_t\). 
    Consequently, \(A\) becomes unusable for locating purposes, and we set \(\calY = \calY' \cup \{A\}\). Otherwise, we retain \(\calY = \calY'\).

    \item Any pair present in exactly one of \(\mathbb{W}\) and \(\mathbb{W}'\) necessarily involves the vertex \(x\). Since \((x, +) \not\in \mathbb{W}'\), and for every pair \((w_1, x) \in \mathbb{W}'\), the vertex \(x\) is forgotten at node \(t\), such a pair must be replaced by \((w_1, +)\) in \(\mathbb{W}\), according to the convention.
\end{enumerate}

Thus, \(S_t = S_{t'}\) constitutes a valid candidate solution for \(\dptw[t, \tau]\), and we have
\(\dptw[t, \tau] \le \dptw[t', \tau']\).

\smallskip

Conversely, suppose that \(S_t\) is an optimal candidate solution at node \(t\) for tuple \(\tau\)
such that \(x \not\in S_t\). 
We now demonstrate that \(S_{t'} := S_t\) is a valid candidate solution at node \(t'\) for the 
corresponding tuple \(\tau'\).

\begin{enumerate}[nolistsep]
    \item Since \(x\) is forgotten in this bag, we have \(V(G_{t'}) \setminus X_{t'} \subseteq V(G_t) \setminus X_t\), and therefore \(S_{t'}\) still locates all vertices in \(V(G_{t'}) \setminus X_{t'}\).

    \item The partial solution remains unchanged, implying \(Y = Y'\). By definition, \(W = W' \setminus \{x\}\), and thus \(W'\) consists of all vertices in \(V(G_t)\) that are dominated by \(S_{t'}\).
    Using similar arguments, \(W_0 = W'_0 \setminus \{x\}\).

    \item Define \(A := S_t \cap N_{G_{t}}(x)\).
    If \(A \subseteq Y\), then \(A\) is marked unusable and \(A \in \calY\). However, at node \(t'\), where \(x\) is included in the bag, the set \(A\) becomes usable, and hence \(\calY' = \calY \setminus \{A\}\).
 Otherwise, the set of unusable sets remains unchanged.

    \item Since \(W' = W \cup \{x\}\), any discrepancy between \(\mathbb{W}\) and \(\mathbb{W}'\) involves the vertex \(x\).
     Suppose that there exists \((w_1, +) \in \mathbb{W}\) such that \(N_{G_t}(x) \cap S_t = A\). 
     Then, at node \(t'\), the vertices \(x\) and \(w_1\) share identical neighborhoods in \(S_t\). 
     As \(S_t\) locates \(V(G_t) \setminus X_t\), \(x\) is the unique vertex whose 
     neighborhood within \(S_t\) is \(A\). 
     Consequently, \(\mathbb{W}'\) is obtained from \(\mathbb{W}\) by replacing each such pair \((w_1, +)\) with \((w_1, x)\). 
\end{enumerate}

Hence, \(S_{t'} = S_t\) is a valid candidate solution at node \(t'\), and we have
\(\dptw[t, \tau] \ge \dptw[t', \tau']\).

\smallskip
By combining both directions, we establish the correctness of the recursive update at forget nodes in the case 
where the forgotten vertex \(x\) is not part of the candidate solution but is dominated 
by some vertex already included in the candidate solution.

\smallskip
\textbf{Case 3: \(x\) is  in \({Z'}\).} 
If \(N_{G_{t'}}[x] \cap Y' \neq \emptyset\), then either \( x \) is in the partial solution or 
adjacent to some vertex in it, violating the definition of set $Y'$ or $W'$. 
In this case, set \(\dptw[t, \tau] := \infty\).
Otherwise, the forgotten vertex \( x \) 
was not adjacent to any vertex of the partial solution, and any extention of it will remain invalid, as $x$ will remain undominated.
Hence, also in this case \(\dptw[t, \tau] := \infty\).

This concludes the proof of the lemma.
\end{proof}

\subparagraph*{Join Node.} 
Let \( t \) be a join node with children 
\( t^1 \) and \( t^2 \), where \( X_t = X_{t^1} = X_{t^2} \), 
and let \( \tau = \langle Y, W, W_0, \calY, \mathbb{W} \rangle \) 
be a valid tuple at \( t \).

{A crucial property of tree-decompositions is that the intersection of any two bags corresponding to adjacent nodes in the decomposition tree forms a separator in $G$ (see for example the book~\cite[Lemma 7.1]{cygan2015parameterized}). Hence, if $V(G_{t^1})\setminus X_t$ and $V(G_{t^2})\setminus X_t$ are both nonempty, then $X_t$ forms a separator between the two corresponding subgraphs. We will use this implicitly in the proofs below.}

We start with some useful definitions.

\begin{definition}[Pairwise Join-Compatible]
\label{def:pairwise-join-compatible}
Consider two tuples 
\( \tau^1 = \langle Y^1, W^1, W^1_0, \calY^1, \mathbb{W}^1 \rangle \) and 
\( \tau^2 = \langle Y^2, W^2, W^2_0, \calY^2, \mathbb{W}^2 \rangle \),
valid at nodes \( t^1 \) and \( t^2 \), respectively.
We say that \(\tau^1\) and \(\tau^2\) are 
\emph{pairwise join-compatible}, denoted by \(\langle \tau^1, \tau^2 \rangle\),
if the following conditions hold:
\begin{enumerate}[nolistsep]
    \item \(Y^1 = Y^2\),
    \item \(\calY^1 \cap \calY^2 = \emptyset\), and
    \item 
      for any \(w \in W^1_0 \cap W^2_0\), the pair \((w, +)\) appears in at most one of the sets 
        \(\mathbb{W}^1\) and \(\mathbb{W}^2\).
\end{enumerate}
\end{definition}

Note that even though \(Y^1 = Y^2\), it may happen that 
\(W^1 \setminus N(Y^1) \neq W^2 \setminus N(Y^2)\),
and thus \(W^1 \neq W^2\), because \(W^1\) and \(W^2\) respectively 
denote the vertices in \(X_{t^1}\) and \(X_{t^2}\) that are dominated by vertices 
in the partial solutions of \(G_{t^1}\) and \(G_{t^2}\). 
Additionally, there may exist a vertex \(w \in W^1_0\) such that \(w \not\in W^2_0\) {(or vice-versa)},
as \(w\) could be adjacent to some vertex in \(V(G_{t^2}) \setminus X_{t^2}\) that is 
part of the partial solution.

\begin{definition}[Join-compatible]
\label{def:join-compatible}
We say that \(\tau\) is \emph{join-compatible} with the pair \(\langle \tau^1, \tau^2 \rangle\) 
of pairwise join-compatible tuples {valid at $t^1$ and $t^2$, respectively,} if the 
following conditions are satisfied:
\begin{enumerate}[nolistsep]
    \item \(Y = Y^1 = Y^2\), \(W = W^1 \cup W^2\), and \(W_0 = W^1_0 \cap W^2_0\),
    \item \(\calY = \calY^1 \cup \calY^2\), and
    \item
    \begin{enumerate}
        \item A pair \((w, w')\) is in \(\mathbb{W}\) if and only one of the following conditions is 
        satisfied.  
      		\begin{itemize}
      		\item If \(w \in W^1 \cap W^2\), then \((w, w') \in \mathbb{W}^1 \cap \mathbb{W}^2\).
      		\item If \(w \in W^1 \setminus W^2\), then \((w, w')\) is in \(\mathbb{W}^1\) and not in \(\mathbb{W}^2\).
		
		(This refers to the case when \(w'\) is adjacent to some vertex in the partial solution in \(G_{t^1}\)
		but not adjacent to any vertex of the partial solution in \(G_{t^2}\).)
		\item If \(w \in W^2 \setminus W^1\), then \((w, w')\) is in \(\mathbb{W}^2\) and not in \(\mathbb{W}^1\).
		\end{itemize}
		Note that for any set  in \(\{ W^1\cap W^2, W^1 \setminus W^2, W^2 \setminus W^1\}\), 
		\(w\) is present in the set if and only if \(w'\) is in it as well. 
		\item A pair \((w, +)\) is in \(\mathbb{W}\) if and only if of the following conditions is 
        satisfied. 
      	\begin{itemize}
      		\item If \(w \in W^1 \cap W^2\), then consider the following subcases:
      			\begin{itemize}
		      		\item If \(w \in W^1_0 \cap W^2_0\), then \((w, +)\) is either in \(\mathbb{W}^1\) or \(\mathbb{W}^2\), 
		      		but not both.
		      		\item If \(w \in W^1_0\) and \(w \in W^2 \setminus W^2_0\), then \((w, +)\)
		      		is in \(\mathbb{W}^2\) (and may be in \(\mathbb{W}^1\)).
      				\item If \(w \in W^2_0\) and \(w \in W^1 \setminus W^1_0\), then \((w, +)\)
		      		is in \(\mathbb{W}^1\) (and may be in \(\mathbb{W}^2\)).
      			\end{itemize}
      		\item If \(w \in W^1 \setminus W^2\), then \((w, +)\) is in \(\mathbb{W}^1\) and not in \(\mathbb{W}^2\).
			\item If \(w \in W^2 \setminus W^1\), then \((w, +)\) is in \(\mathbb{W}^2\) and not in \(\mathbb{W}^1\).
		\end{itemize}
	\end{enumerate}
\end{enumerate}
Let \( J(\tau) \) denote the collection of pairs \(\langle \tau^1, \tau^2 \rangle\) 
that are join-compatible with \(\tau\).
\end{definition}

\begin{lemma}
\label{lemma:join-node}
Let \( t \) be a join node with children 
\( t^1, t^2 \), where \( X_t = X_{t^1} = X_{t^2} \), 
and let \( \tau = \langle Y, W, W_0, \calY, \mathbb{W} \rangle \) be a valid tuple at \( t \). Then,
\[
\dptw[t, \tau] = \min\limits_{\langle \tau^1, \tau^2 \rangle \in J(\tau)} \{ \dptw[t^1, \tau^1] + \dptw[t^2, \tau^2]\}  - |Y|. 
\]
\end{lemma}
\begin{proof}
Suppose that the minimum value of 
\(
\min_{\langle \tau^1, \tau^2 \rangle \in J(\tau)} \left\{ \dptw[t^1, \tau^1] + \dptw[t^2, \tau^2] \right\}
\)
is finite and attained at the pair \((\tau^1, \tau^2)\).  
Let \(S_{t^1}\) and \(S_{t^2}\) be the corresponding partial solutions that realize this minimum.  
Define 
\(
S_t := S_{t^1} \cup S_{t^2}.
\)
We now verify that \(S_t\) is a valid candidate solution at node \(t\) by checking that it satisfies the requirements of Definition~\ref{def:candidate-sol}.  
The relationships among the sets specified by \(\tau\), \(\tau^1\), and \(\tau^2\) are governed by Definition~\ref{def:pairwise-join-compatible} and Definition~\ref{def:join-compatible}.

\begin{enumerate}[nolistsep]
  \item By definition, \(S_{t^1}\) and \(S_{t^2}\) locate all vertices in \(V(G_{t^1}) \setminus X_t\) and \(V(G_{t^2}) \setminus X_t\), respectively.  
  Since \(V(G_{t^1}) \setminus X_t\) and \(V(G_{t^2}) \setminus X_t\) are disjoint, and only vertices in \(X_t\) can have neighbors in both subgraphs, it follows that \(S_t\) locates any \(u \in V(G_t) \setminus X_t\) provided \(N_{G_t}(u) \cap S_t\) is not entirely contained in \(X_t\).  
  Moreover, as \(\tau^1\) and \(\tau^2\) are pairwise compatible, we have \(\mathcal{Y}^1 \cap \mathcal{Y}^2 = \emptyset\).  
  Equivalently, for any \(A \subseteq Y\), there exists at most one vertex in \(V(G_t) \setminus X_t\) with neighborhood \(A\) in \(G_t\).  
  Therefore, \(S_t\) indeed locates every vertex in \(V(G_t) \setminus X_t\).

  \item Since \(Y = Y^1 = Y^2\), the intersection of \(S_t\) with \(X_t\) is precisely \(Y\).  
  Furthermore, \(W = W^1 \cup W^2\) denotes exactly the vertices of \(X_t\) dominated by \(S_t\), and  
  \(W_0 = W_0^1 \cap W_0^2\) consists of those vertices whose neighbors in \(S_t\) lie entirely within \(Y\).

  \item No vertex in \(V(G_{t^1}) \setminus X_t\) is adjacent to any vertex in \(V(G_{t^2}) \setminus X_t\).  
  This implies that no vertex in \(V(G_{t^1}) \setminus (S_{t^1} \cup X_t)\) is adjacent with a 
  vertex in \(S_{t^2} \setminus X_t\).
  A symmetric statement holds for vertices in \(V(G_{t^2}) \setminus (S_{t^2} \cup X_t)\).
  Consequently, the set of unusable subsets at \(t\) is given by  
  \(
  \mathcal{Y} = \mathcal{Y}^1 \cup \mathcal{Y}^2.
  \)

  \item It remains to verify that the set of unresolved pairs in \(X_t\) with respect to \(S_t\) are given by 
  the set \(\mathbb{W}\), as prescribed in Definition~\ref{def:join-compatible}.
  \begin{enumerate}
    \item 
    Consider \((w, w') \in \mathbb{W}\) with \(w, w' \in W\).  
    This pair is unresolved by \(S_t\) if and only if \(N_{G_t}(w) \cap S_t = N_{G_t}(w') \cap S_t\).  
    For any of the three sets \(\{ W^1 \cap W^2,\; W^1 \setminus W^2,\; W^2 \setminus W^1 \}\),  
    the membership of \(w\) and \(w'\) is identical.

    \begin{itemize}
      \item If \(w \in W^1 \cap W^2\), let \(N_i(v) := N_{G_{t^i}}(v) \cap S_{t^i}\) for \(i \in \{1,2\}\).  
      Then \(N_{G_t}(v) \cap S_t = N_1(v) \cup N_2(v)\), and the disjointness of  
      \(S_{t^1} \setminus X_t\) and \(S_{t^2} \setminus X_t\) implies  
      \(
      N_1(w) \cup N_2(w) = N_1(w') \cup N_2(w')\) if and only if \(N_1(w) = N_1(w') \ \text{and} \ N_2(w) = N_2(w').
      \)
      Thus, \((w, w')\) is unresolved by \(S_t\) if and only if it is unresolved in both \(S_{t^1}\) and \(S_{t^2}\),  
      i.e., \((w, w') \in \mathbb{W}^1 \cap \mathbb{W}^2\).

      \item If \(w \in W^1 \setminus W^2\), then \((w, w') \notin \mathbb{W}^2\), implying it is unresolved solely in \(S_{t^1}\), hence \((w, w') \in \mathbb{W}^1\).

      \item If \(w \in W^2 \setminus W^1\), the argument is symmetric to the previous case.
    \end{itemize}

    \item  
    For a vertex \(w \in W\), \((w, +) \in \mathbb{W}\) if \(w\) shares the same neighborhood in 
    \(S_t\) as some \(u \in V(G_t) \setminus (S_t \cup X_t)\).  
    Let \(u_1\) (resp. \(u_2\)) be such a vertex in \(V(G_{t^1}) \setminus (S_{t^1} \cup X_t)\) (resp. \(V(G_{t^2}) \setminus (S_{t^2} \cup X_t)\)), if it exists.  
    Define \(A_1 := N_{G_{t^1}}(w) \cap S_{t^1}\) and \(A_2 := N_{G_{t^2}}(w) \cap S_{t^2}\).

    \begin{itemize}
      \item If \(w \in W^1 \cap W^2\):
      \begin{itemize}
        \item If \(w \in W_0^1 \cap W_0^2\), then both \(u_1\) and \(u_2\) can not exists, else it contradicts that 
        \(S_t\) locates all vertices in \(V(G_t) \setminus X_t\).  
        Hence, \((w, +)\) lies in exactly one of \(\mathbb{W}^1\) or \(\mathbb{W}^2\), but not both.
        \item If \(w \in W^1_0\) and \(w \in W^2 \setminus W^2_0\).
		      then \(A_1\) is in \(Y\) and is a proper subset of \(A_2\).
		      Hence, a vertex in \(A_2 \setminus A_1\) resolves the pair \(w, u_1\), if it exits.
		      This implies that \((w, +)\) is in \(\mathbb{W}^2\).
		      Note that  \((w, +)\) is in \(\mathbb{W}^1\) if and only if \(u_1\) exists.
        \item If \(w \in W_0^2\) and \(w \in W^1 \setminus W_0^1\), the argument is symmetric.
        \item If \(w \in W^1 \setminus W^1_0\) and \(w \in W^2 \setminus W^2_0\),
        then a vertex in \(A_1 \setminus A_2\) resolves the pair \((w, u_2)\) 
      		and a vertex in \(A_2 \setminus A_1\) resolves the pair \((w, u_1)\).
      		This implies that \((w, +)\) is \emph{not} present in \(\mathbb{W}\).
      \end{itemize}

      \item Suppose that \(w \in W^1 \setminus W^2\).
      	In this case, pair \((w, +)\) is not present in \( W^2\).
       	Hence, \(u = u_1\) which implies that \((w, +)\) is in \(\mathbb{W}^1\).
		\item If \(w \in W^2 \setminus W^1\), then \((w, +)\) is in \(\mathbb{W}^2\) (only and not in \(\mathbb{W}^1\)).
		This case follows using symmetric arguments as in the previous case.
    \end{itemize}
  \end{enumerate}
\end{enumerate}

\smallskip
Thus, \(S_t\) is a valid candidate solution corresponding to the tuple \(\tau\).  
Its size is given by 
\(
|S_t| = |S_{t^1}| + |S_{t^2}| - |Y|.
\)
Minimizing over all join-compatible pairs \(\langle \tau^1, \tau^2 \rangle \in J(\tau)\) yields
\(
\dptw[t, \tau] \leq \min_{\langle \tau^1, \tau^2 \rangle \in J(\tau)} \left\{ \dptw[t^1, \tau^1] + \dptw[t^2, \tau^2] \right\} - |Y|.
\)

\smallskip

Conversely, let \(S_t\) be an optimal candidate solution for the tuple \(\tau\) at node \(t\). 
By definition, \(|S_t| = \dptw[t, \tau]\), and \(S_t\) satisfies all conditions specified in Definition~\ref{def:candidate-sol} for \(\tau\). 
Define
\(
S_{t^1} := S_t \cap V(G_{t^1})\),  and \(S_{t^2} := S_t \cap V(G_{t^2}).
\)
We first establish certain properties of \(S_{t^1}\) and \(S_{t^2}\), and then define the tuples
\(
\tau_1 := \langle Y^1, W^1, W_0^1, \mathcal{Y}^1, \mathbb{W}^1 \rangle\), 
and 
\(\tau_2 := \langle Y^2, W^2, W_0^2, \mathcal{Y}^2, \mathbb{W}^2 \rangle.
\)

\begin{enumerate}[nolistsep]
  \item \(S_{t^1}\) locates (in \(G_{t^1}\)) all vertices in \(V(G_{t^1}) \setminus X_{t^1}\), and 
        \(S_{t^2}\) locates (in \(G_{t^2}\)) all vertices in \(V(G_{t^2}) \setminus X_{t^2}\).  
        This follows since \(S_t\) locates all vertices in \(V(G_{t^1})\), and no vertex in \(V(G_{t^1}) \setminus X_{t^1}\) is adjacent to any vertex in \(V(G_{t^2}) \setminus X_{t^2}\).

  \item Define \(Y^1 := S_{t^1} \cap X_{t^1}\).  
        Let \(W^1\) be the set of vertices in \(X_{t^1} \setminus Y^1\) that are dominated by \(S_{t^1}\).  
        Let \(W_0^1 \subseteq W^1\) be the set of vertices whose neighborhood in \(S_{t^1}\) is contained in \(Y^1\).  
        Define \(Y^2, W^2, W_0^2\) analogously for \(t^2\).

  \item Define \(\mathcal{Y}^1\) as the collection of subsets \(A \subseteq Y\) such that there exists a vertex 
        \(u \in V(G_{t^1}) \setminus (S_{t^1} \cup X_{t^1})\) with
        \(
        N_{G_{t^1}}(u) \cap S_t = A.
        \)
        Since \(S_{t^1}\) locates every vertex in \(V(G_{t^1}) \setminus X_{t^1}\), such a vertex \(u\) (if it exists) is unique.  
        Define \(\mathcal{Y}^2\) analogously.

  \item 
        \begin{itemize}
          \item For \(w, w' \in W^1\), add the pair \((w, w')\) to \(\mathbb{W}^1\) if and only if
                \(N_{G_{t^1}}(w) \cap S_{t^1} = N_{G_{t^1}}(w') \cap S_{t^1}\).
          \item For \(w \in W^1\), add the pair \((w, +)\) to \(\mathbb{W}^1\) if and only if there exists 
                \(u \in V(G_{t^1}) \setminus (S_{t^1} \cup X_{t^1})\) such that
                \(N_{G_{t^1}}(w) \cap S_t = N_{G_{t^1}}(u) \cap S_{t^1}\).
        \end{itemize}
        Define \(\mathbb{W}^2\) analogously.
\end{enumerate}

By Definition~\ref{def:valid-tuple}, 
\(\tau_1\) and \(\tau_2\) are valid tuples at \(t^1\) and \(t^2\), respectively.  
By Definition~\ref{def:candidate-sol},
\(S_{t^1}\) and \(S_{t^2}\) are candidate solutions at \(t^1\) with respect to \(\tau_1\), and at \(t^2\) with respect to \(\tau_2\), respectively.  

Next, we prove that \(\tau_1\) and \(\tau_2\) are pairwise join-compatible by showing that they satisfy the properties in Definition~\ref{def:pairwise-join-compatible}:

\begin{enumerate}[nolistsep]
    \item By construction, \(Y^1 = Y^2\).
    \item Suppose that \(\mathcal{Y}^1 \cap \mathcal{Y}^2 \neq \emptyset\).  
          Then there exist vertices \(u_1 \in V(G_{t^1}) \setminus X_{t^1}\) and \(u_2 \in V(G_{t^2}) \setminus X_{t^2}\) such that their neighborhoods in \(S_t\) are identical and contained in \(Y\).  
          This contradicts the fact that \(S_t\) locates every vertex in \(V(G_t) \setminus X_t\).  
          Therefore, \(\mathcal{Y}^1 \cap \mathcal{Y}^2 = \emptyset\).
    \item Suppose that there exists \(w \in W_0^1 \cap W_0^2\) such that \((w, +) \in \mathbb{W}^1 \cap \mathbb{W}^2\).  
          Then there exist \(u_1 \in V(G_{t^1}) \setminus X_{t^1}\) and \(u_2 \in V(G_{t^2}) \setminus X_{t^2}\) whose neighborhoods in \(S_t\) are identical to that of \(w\) and contained in \(Y\).  
          This again contradicts the fact that \(S_t\) locates every vertex in \(V(G_t) \setminus X_t\).
\end{enumerate}
Hence, \(\tau_1\) and \(\tau_2\) are pairwise join-compatible.

Next, we establish that \(\langle \tau^1, \tau^2 \rangle\) is join-compatible with \(\tau\) by verifying that it satisfies all the properties stated in Definition~\ref{def:join-compatible}.

\begin{enumerate}[nolistsep]
    \item By definition, \(Y = Y^1 = Y^2\).

    For any vertex \(w \in V(G_t)\), we have
    \(
        N_{G_t}[w] \cap S_t \;=\; (N_{G_{t^1}}[w] \cap S_{t^1}) \;\cup\; (N_{G_{t^2}}[w] \cap S_{t^2}).
    \)
    Consequently, \(w\) is dominated by \(S_t\) if and only if it is dominated by either \(S_{t^1}\) or \(S_{t^2}\).  
    By definition, \(W^1\) and \(W^2\) are precisely the sets of vertices in \(X_t\) dominated by \(S_{t^1}\) and \(S_{t^2}\), respectively. Therefore,
    \(
        W = W^1 \cup W^2.
    \)

    A vertex \(w \in X_t\) belongs to \(W_0\) if all its neighbors in \(S_t\) are contained in \(Y\), i.e.,
    \(
        N_{G_t}(w) \cap S_t \subseteq Y.
    \)
    This holds if and only if both
    \(N_{G_{t^1}}(w) \cap S_{t^1} \subseteq Y\)
{and}
        \(N_{G_{t^2}}(w) \cap S_{t^2} \subseteq Y\),
    which are precisely the conditions for \(w \in W_0^1\) and \(w \in W_0^2\), respectively.  
    Thus,
    \(
        W_0 = W_0^1 \cap W_0^2.
    \)

    \item Let \(A \in \calY\), and suppose that \(u \in V(G_t) \setminus (S_t \cup X_t)\) is such that \(N_{G_t}(u) \cap S_t = A\).  
    The set of such vertices \(u\) is the disjoint union of:
    \begin{itemize}
        \item those in \(V(G_{t^1}) \setminus (S_{t^1} \cup X_t)\), for which \(N_{G_t}(u) \cap S_t = N_{G_{t^1}}(u) \cap S_{t^1} \in \calY^1\);
        \item those in \(V(G_{t^2}) \setminus (S_{t^2} \cup X_t)\), for which \(N_{G_t}(u) \cap S_t = N_{G_{t^2}}(u) \cap S_{t^2} \in \calY^2\).
    \end{itemize}
    Therefore,
    \(
        \calY = \calY^1 \cup \calY^2.
    \)

    \item
    \begin{enumerate}
        \item Since \(W = W^1 \cup W^2\), consider any \(w, w' \in W\). If \((w, w') \in \mathbb{W}\), then for each set in \(\{ W^1 \cap W^2,\, W^1 \setminus W^2,\, W^2 \setminus W^1 \}\), the membership of \(w\) coincides with that of \(w'\).

        \begin{itemize}
            \item If \(w, w' \in W^1 \cap W^2\), then \((w, w') \in \mathbb{W}\) if and only if
            \(
                (w, w') \in \mathbb{W}^1 \cap \mathbb{W}^2.
            \)
            Indeed, \((w, w') \in \mathbb{W}\) means \(N_{G_t}(w) \cap S_t = N_{G_t}(w') \cap S_t\), which is equivalent to
            \(
                N_{G_{t^1}}(w) \cap S_{t^1} = N_{G_{t^1}}(w') \cap S_{t^1}\)
                {and} 
                \(N_{G_{t^2}}(w) \cap S_{t^2} = N_{G_{t^2}}(w') \cap S_{t^2}.
            \)

            \item If \(w, w' \in W^1 \setminus W^2\), then \((w, w') \in \mathbb{W}\) if and only if \((w, w') \in \mathbb{W}^1\).

            \item If \(w, w' \in W^2 \setminus W^1\), the statement follows symmetrically from the previous case.
        \end{itemize}

        \item Let \(w \in W\), and let \(u_1\) (resp. \(u_2\)) be a vertex, if it exists, in \(V(G_{t^1}) \setminus (S_{t^1} \cup X_t)\) (resp. \(V(G_{t^2}) \setminus (S_{t^2} \cup X_t)\)) whose neighborhood in \(S_{t^1}\) (resp. \(S_{t^2}\)) matches that of \(w\). Denote these neighborhoods by \(A_1\) and \(A_2\), respectively.  
        Similarly, let \(u\) (if it exists) be the vertex in \(V(G_t) \setminus (S_t \cup X_t)\) whose neighborhood in \(S_t\) matches that of \(w\). By construction, \(u\) is either \(u_1\) or \(u_2\).
	
		\begin{itemize}
      	\item Suppose that \(w \in W^1 \cap W^2\), then we consider the following subcases:
      	\begin{itemize}
		      \item Suppose that \(w \in W^1_0 \cap W^2_0\).
		      Vertex \(u\) exists if and only if exactly one of \(u_1\) and \(u_2\) exists as otherwise 
		      it contradicts the fact that \(S_t\) locates in \(G_t\) every vertex in \(V(G_t) \setminus X_t\).
  			 Hence, \((w, +) \in \mathbb{W}\) if and only if 
  			 \((w, +)\) in either in \(\mathbb{W}^1\) or \(\mathbb{W}^2\) 
		     (but not both).
		      \item Suppose that \(w \in W^1_0\) and \(w \in W^2 \setminus W^2_0\).
		      This implies \(A_1\) is in \(Y\) and is a proper subset of \(A_2\).
		      Hence, a vertex in \(A_2 \setminus A_1\) resolve the pair \(w, u_1\).
		      This implies \(u\) exists if and only if \(u_2\) exits.
		      Hence, \((w, +)\) is in \(\mathbb{W}\) if and only if \((w, +)\) is in \(\mathbb{W}^2\).
		      Note that depending on whether \(u_1\) exists or not, \((w, +)\) may be present in 
		      \(\mathbb{W}^1\).      		
      		\item Suppose that \(w \in W^2_0\) and \(w \in W^1 \setminus W^1_0\).
      		This case follows using symmetric arguments as in the previous case.
      		\item Suppose that \(w \in W^1 \setminus W^1_0\) and \(w \in W^2 \setminus W^2_0\).
      		In this case, a vertex in \(A_1 \setminus A_2\) resolves the pair \(w, u_2\) 
      		and a vertex in \(A_2 \setminus A_1\) resolves the pair \(w, u_1\).
      		This implies that \(w, +\) is \emph{not} present in \(\mathbb{W}\).
      	\end{itemize}
      	\item Suppose that \(w \in W^1 \setminus W^2\).
      	In this case, \(u_2\) does not exist, and hence \(u\) exists if and only if \(u_1\) exits.
      	Hence \((w, +)\) is in \(\mathcal{W}\) if and only if \((w, +)\) is in \(\mathbb{W}^1\).
		\item Suppose that \(w \in W^2 \setminus W^1\).
		This case follows using symmetric arguments as in the previous case.
  \end{itemize}

    \end{enumerate}
\end{enumerate}

From the above, it follows that \(\tau\) is join-compatible with the pair \(\langle \tau^1, \tau^2 \rangle\).

\medskip
In summary, \(S_{t^1}\) and \(S_{t^2}\) are valid candidate solutions for \(\tau^1\) and \(\tau^2\), respectively, and \(\langle \tau^1, \tau^2 \rangle\) is join-compatible with \(\tau\). Moreover,
\(
    |S_t| = |S_{t^1} \cup S_{t^2}| = |S_{t^1}| + |S_{t^2}| - |S_{t^1} \cap S_{t^2}| = |S_{t^1}| + |S_{t^2}| - |Y|.
\)
Since \(\dptw[t^1, \tau^1]\) and \(\dptw[t^2, \tau^2]\) denote the sizes of a minimum candidate solutions for \(\tau^1\) and \(\tau^2\), respectively, we have \(|S_{t^1}| \geq \dptw[t^1, \tau^1]\) and \(|S_{t^2}| \geq \dptw[t^2, \tau^2]\). Substituting into the above expression yields
\(
    \dptw[t, \tau] = |S_t| \;\geq\; \dptw[t^1, \tau^1] + \dptw[t^2, \tau^2] - |Y|.
\)
This inequality holds for any \(\langle \tau^1, \tau^2 \rangle \in J(\tau)\), where \(J(\tau)\) denotes the set of all pairs satisfying Definition~\ref{def:join-compatible}. Therefore,
\(
    \dptw[t, \tau] \;\geq\; \min_{\langle \tau'^1, \tau'^2 \rangle \in J(\tau)}
    \big\{ \dptw[t^1, \tau'^1] + \dptw[t^2, \tau'^2] \big\} - |Y|.
\)
This completes the proof of the reverse inequality.  
By combining both directions, we conclude the correctness of the recursive update at join nodes.
\end{proof}

We are now in a position to present the proof of Theorem~\ref{thm:LD-tw-algo}
which states that \LD\ admits an algorithm with running time
$2^{2^{\calO(\tw)}} \cdot n$.

\begin{proof}[Proof of Theorem~\ref{thm:LD-tw-algo}]
The algorithm first computes a nice tree-decomposition of width $\calO(\tw)$ with $\calO(n)$ nodes
in time $2^{\calO(\tw)} \cdot n$~\cite{niceTW,K21}.
It then evaluates the dynamic programming table in a bottom-up manner, starting from 
the leaf nodes 
of the decomposition and proceeding towards the root node \(r\).  
Finally, it outputs the set corresponding to \(\dptw[r, \tau]\), 
where \(\tau = \langle \emptyset, \emptyset, \emptyset, \emptyset, \emptyset \rangle\).

The fact that each DP-state can be computed correctly follows from 
Lemmas~\ref{lemma:leaf-node}, \ref{lemma:introduce-node}, 
\ref{lemma:forget-node}, and \ref{lemma:join-node}.
Together with the properties of nice tree-decompositions, 
these lemmas imply that the set returned by the algorithm is 
indeed an optimal locating-dominating set of \(G\).
This establishes the correctness of the algorithm.

We now analyse the running time.  
The total number of possible DP-states at a node is bounded by 
\(2^{2^{\calO(\tw)}}\), and all other operations at a node 
can be performed in time polynomial in the number of states.  
Therefore, since there are $\calO(n)$ nodes, the overall running time is 
\(2^{2^{\calO(\tw)}} \cdot n\), as claimed.
This concludes the proof of Theorem~\ref{thm:TC-tw-algo}.
\end{proof}

\subsection{Lower Bound}
\label{sub-sec:locating-dom-tw-lb}
In this subsection, we prove the lower bound for \LD
mentioned in Theorem~\ref{thm:TW}. 
For convenience, we restate this part of the statement as follows.

\begin{theorem*}[Restating part of Theorem~\ref{thm:TW}]
Unless the \ETH\ fails,
\LD does not admit an algorithm 
running in time  $2^{2^{o(\tw)}} \cdot {\rm poly}(n)$, where $\tw$ is the treewidth and $n$ is the order of the input graph.
\end{theorem*}

To prove the above theorem, we present a reduction 
from a variant of \textsc{$3$-SAT} called
\textsc{$(3,3)$-SAT}.
In this variation, an input is a boolean satisfiability
formula $\psi$ in conjunctive normal form 
such that each clause contains \emph{at most}\footnote{We remark 
that if each clause contains \emph{exactly} $3$ variables,
and each variable appears $3$ times, then the problem
is polynomial-time solvable \cite[Theorem 2.4]{DBLP:journals/dam/Tovey84}}
$3$ variables, and each variable appears at most $3$ times.
Consider the following reduction from an instance $\phi$
of \textsc{$3$-SAT} with $n$ variables and $m$ clauses
to an instance $\psi$ of \textsc{$(3, 3)$-SAT} mentioned in 
\cite{DBLP:journals/dam/Tovey84}:
For every variable $x_i$ that appears $k > 3$ times,
the reduction creates $k$ many new variables
$x^1_i, x^2_i, \dots, x^{k}_i$,
replaces the $j^{th}$ occurrence of $x_i$ by $x^{j}_i$, 
and adds the series of new clauses to encode
$x^1_i \Rightarrow x^2_i \Rightarrow \cdots \Rightarrow x^k_i 
\Rightarrow x^1_i$.
For an instance $\psi$ of \textsc{$3$-SAT},
suppose $k_i$ denotes the number of times a variable $x_i$
appeared in $\phi$.
Then, $\sum_{i \in [n]} k_i \le 3 \cdot m$.
Hence, the reduced instance
$\psi$ of \textsc{$(3,3)$-SAT} has at most $3m$ variables
and $4m$ clauses.
Using the \ETH~\cite{DBLP:journals/jcss/ImpagliazzoP01} and 
the sparsification lemma~\cite{DBLP:journals/jcss/ImpagliazzoPZ01}, 
we have the following result.

\begin{proposition}
\label{prop:3-3-SAT-ETH-lb}
\textsc{$(3,3)$-SAT}, with $n$ variables and $m$ clauses, 
does not admit an algorithm running in time $2^{o(m+n)}$,
unless the \emph{\ETH} fails.
\end{proposition}

We highlight that every variable appears positively and negatively
at least once.
Otherwise, if a variable appears only positively (respectively, only negatively)
then we can assign it \true\ (respectively, \false) 
and safely reduce the instance by removing the clauses 
containing this variable.
Hence, instead of the first, second, or third appearance of the variable,
we use the first positive, first negative, second positive, or
second negative appearance of the variable.  

\subparagraph*{Reduction.}
The reduction takes as input an instance $\psi$ of 
\textsc{$(3,3)$-SAT} with $n$ variables and 
outputs an instance $(G,k)$ of \LD\ such that $\tw(G) = \calO(\log(n))$. 
Suppose $X = \{x_1, \dots, x_n\}$ is the collection 
of variables and $C = \{C_1, \dots, C_m\}$ is the collection
of clauses in $\psi$.
Here, we consider 
$\langle x_1, \dots, x_n\rangle$ and
$\langle C_1, \dots, C_m \rangle$ to be arbitrary but fixed
orderings of variables and clauses in $\psi$.
For a particular clause, the first order specifies
the first, second, or third (if it exists) variable in the clause 
in a natural way.
The second ordering specifies the first/second positive/negative
appearance of variables in $X$ in a natural way.

We will make use of \emph{set-representation gadgets} as defined in~\cite{FoucaudGK0IST24}. Intuituvely speaking, such gadgets form a logarithmic-size separator that allows to connect two linear-size sets of vertices and encode their interactions.

The reduction constructs a graph $G$ as follows.
\begin{itemize}
\item 
To construct a variable gadget for $x_i$,
it starts with two claws $\{\alpha_i^0, \alpha_i^1, \alpha_i^2, \alpha_i^3\}$
and $\{\beta_i^0, \beta_i^1, \beta_i^2, \beta_i^3\}$ centered at
$\alpha_i^0$ and $\beta_i^0$, respectively. 
(Recall that a claw is the star $K_{1, 3}$.) 
It then adds four vertices
$x_i^1, \neg x_i^1, x_i^2, \neg x_i^2$,
and the corresponding edges, as shown in 
Figure~\ref{fig:locating-dom-tw-lb}.
Let $A_i$ be the collection of these twelve vertices and $A = \bigcup_{i=1}^n A_i$.
Define $X_i := \{x_i^1, \neg x_i^1, x_i^2, \neg x_i^2\}$.

\item To construct a clause gadget for $C_j$, the reduction
starts with a star graph centered at $\gamma_j^0$ and
with four leaves $\{\gamma_j^1, \gamma_j^2, \gamma_j^3, \gamma_j^4\}$.
It then adds three vertices $c^1_j, c^2_j, c^3_j$ and 
the corresponding edges shown in Figure~\ref{fig:locating-dom-tw-lb}. 
Let $B_j$ be the collection of these eight vertices and define $B = \bigcup\limits_{j=1}^m B_j$.

\item Let $p$ be the smallest positive integer such that $4n \leq 
\binom{2p}{p}$. 
Define $\calF_p$ as the collection of subsets of $[2p]$ that contains 
exactly $p$ integers (such a collection $\calF_p$ is called a \emph{Sperner 
family}). 
Define $\setrep: \bigcup_{i=1}^n X_i \to 
\calF_{p}$ 
as an injective function by arbitrarily assigning a set in $\calF_p$ to 
a vertex $x^{\ell}_i$ or $\neg x^{\ell}_i$, for every $i \in [n]$ and 
$\ell \in [2]$. 
In other words, every appearance 
of a literal is assigned a distinct subset in $\calF_p$.
\item The reduction adds a \emph{connection portal} $V$, which is a clique 
on $2p$ vertices $v_1, v_2,  \dots, v_{2p}$. For every vertex $v_q$ in $V$, 
the reduction adds a pendant vertex $u_q$ adjacent to $v_q$.
\item For each vertex $x_i^{\ell} \in X$ where $i \in [n]$ and $\ell \in [2]$, 
the reduction adds edges $(x^{\ell}_i, v_q)$ for every 
$q \in \setrep(x^{\ell}_i)$.
Similarly, it adds edges $(\neg x^{\ell}_i, v_q)$ for every 
$q \in \setrep(\neg x^{\ell}_i)$.

\item For a clause $C_j$, 
suppose variable $x_i$ appears positively for the $\ell^{th}$ time 
as the $r^{th}$ variable in $C_j$.
Then, the reduction adds edges across $B$ and $V$ such that the vertices $c_j^r$ and $x^{\ell}_i$ have the 
same neighborhood in $V$, namely, the set 
$\{v_q : q \in \setrep(x^{\ell}_i) \}$.
For example, $x_i$ appears positively for the second time
as the third variable in $C_j$.
Then, the vertices $c_j^{3}$ and $x_{\ell}^2$ should have the same neighborhood in $V$.
Similarly, it adds edges for the negative appearance of the variables.
\end{itemize}

\begin{figure}[t]
\centering
\includegraphics[scale=0.65]{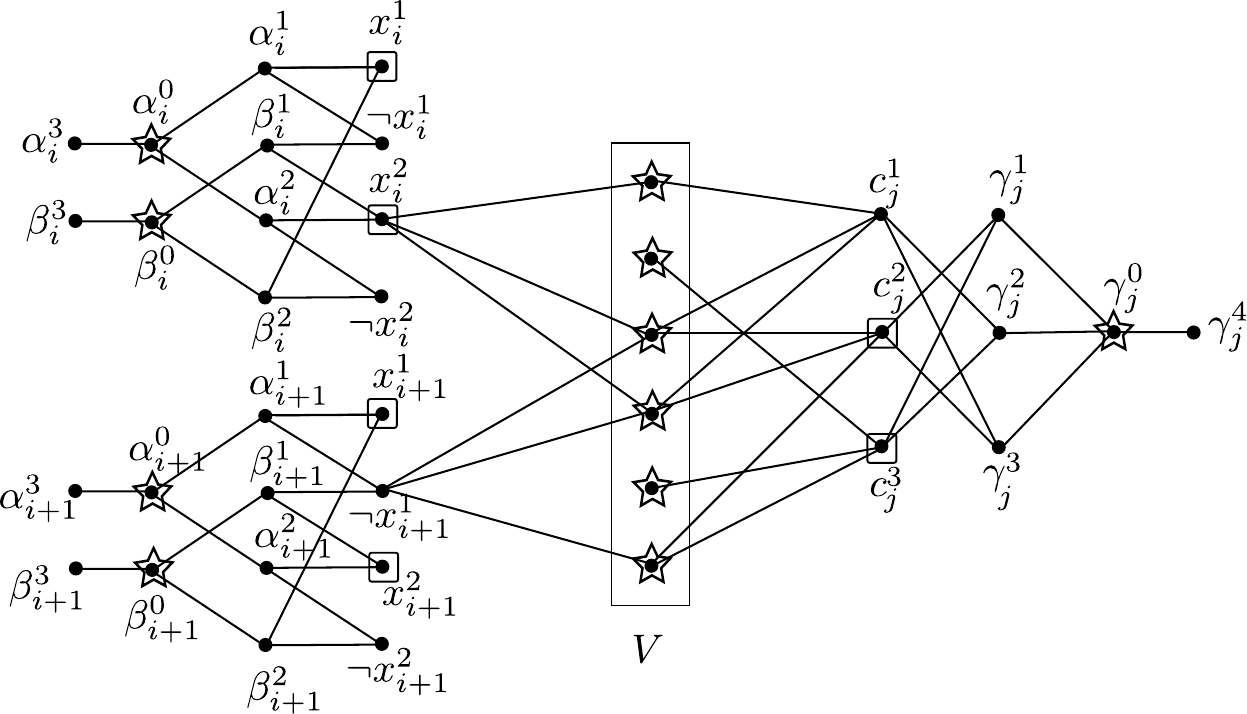}
\caption{Illustration of the reduction used in Subsection~\ref{sub-sec:locating-dom-tw-lb}.
For the sake of clarity, we do not explicitly show the pendant 
vertices adjacent to vertices in $V$.
The variable and clause gadgets are on the left-side and right-side
of $V$, respectively.
In this example, we consider a clause 
$C_j = x_i \lor \neg x_{i + 1} \lor x_{i+2}$.
Moreover, suppose this is the second positive appearance 
of $x_i$ and the first negative appearance of $x_{i + 1}$,
and $x_i$ corresponds to $c^1_j$ and $x_{i+1}$ corresponds 
to $c^2_j$.
Suppose $V$ contains $6$ vertices indexed from top to bottom,
and the set corresponding to these two appearances are 
$\{1, 3, 4\}$ and $\{3, 4, 6\}$ respectively. 
Vertices with a star boundary are those that we can assume to be
in any locating-dominating set, without loss of generality.
The square boundaries correspond to the selection of other vertices in the solution.
In the above example, it corresponds to setting both $x_i$ and $x_{i+1}$
to \true.
On the clause side, the selection corresponds to selecting $x_i$  
to satisfy the clause $C_j$.
\label{fig:locating-dom-tw-lb}}
\end{figure}

This concludes the construction of $G$. 
The reduction sets $k = 4n + 3m + 2p$ and 
returns $(G, k)$ as the reduced instance of {\LD}.

We now prove the validity of the reduction in the following lemmas.

\begin{lemma}
\label{lem:forward-3,3-sat-to-LD}
If $\psi$ be a \yes-instance of \textsc{$(3,3)$-SAT},
then $(G, k)$ is a \yes-instance of \LD.
\end{lemma}

\begin{proof}
Suppose $\pi: X \mapsto \{\true, \false\}$ is a satisfying assignment
of $\psi$.
We construct a vertex subset $S$ of $G$ from the 
satisfying assignment in the following manner:
Initialize $S$ by adding all the vertices in $V$.
For variable $x_i$, if $\pi(x_i) = \true$,
then include $\{\alpha_i^0, \beta_i^0, x^1_i, x^2_i\}$ in $S$
otherwise include $\{\alpha_i^0, \beta_i^0, \neg x^1_i, \neg x^2_i\}$
in $S$.
For any clause $C_j$, 
if its first variable is set to \true\ then include 
$\{\gamma^0_j, c^2_j, c^3_j\}$
in $S$,
if its second variable is set to \true\ then include 
$\{\gamma^0_j, c^1_j, c^3_j\}$ in $S$, otherwise
include $\{\gamma^0_j, c^1_j, c^2_j\}$ in $S$.
If more than one variable of a clause $C_j$ is set to \true,
we include the vertices corresponding to the smallest indexed variable
set to \true.
This concludes the construction of $S$.

It is easy to verify that $|S| = 4n+3m+2p = k$. 
In the remainder of this proof, we argue that $S$ is 
a locating-dominating set of $G$. 
To do so, we first show that $S$ is a dominating set of $G$. 
Notice that $V \subseteq S$ dominates the pendant vertices $u_q$ for all $q \in [2p]$ and all the vertices of the form $x^{\ell}_i$
for any $i \in [n]$ and $\ell \in [2]$,
and $c^{r}_{j}$ for any $j \in [m]$ and $r \in [3]$.
Moreover, the vertices $\alpha_i^0$, $\beta_i^0$ and $\gamma^i_0$ dominate the sets 
$\{\alpha_i^1,\alpha_i^2,\alpha_i^3\}$, 
$\{\beta_i^1,\beta_i^2,\beta_i^3\}$ and 
$\{\gamma^{1}_j,\gamma_j^{2},\gamma_j^{3},\gamma_j^4\}$, 
respectively. 
This proves that $S$ is a dominating set of $G$. 

We now show that $S$ is also a locating set of $G$. 
To begin with, we notice that, for $p \geq 2$, all the pendant vertices $u_q$ 
for $q \in [2p]$ are located from every other vertex in $G$ by the fact that 
$|N_G(u) 
\cap S| = 1$. Next, we divide the analysis of $S$ being a locating set 
of $G$ into the following three cases.

\begin{itemize}
\item First, consider the vertices within $A_i$'s.
Each $x^{\ell}_i$ or $\neg x^{\ell}_i$ for any literal $x_i \in 
X$ and $\ell \in [2]$ have a distinct neighborhood in $V$; hence, they are all pairwise located. 
For any $i \neq i' \in [n]$ and $\ell \neq \ell' \in [2]$, 
the pair $(\alpha_i^{\ell},\beta_{i'}^{\ell'})$ is located by $\alpha_i^0$. 
Moreover, the pair $(\alpha_i^{\ell},\alpha_i^{\ell'})$, 
for all $i \in [n]$, $\ell \neq \ell' \in [2]$ are located by either
$\{x^1_i, x^2_i\}$ or $\{\neg x^1_i, \neg x^2_i\}$, 
one of which is a subset of $S$.

\item Second, consider the vertices within $B_j$'s.
The set of vertices in $\{\gamma_j^{0}, \gamma_j^{1}, \gamma_j^{2}, 
\gamma_j^{3}, \gamma_j^4\}$ 
are pairwise located by three vertices in the set included in $S$. 
One of this vertex is $\gamma_j^{0}$ and the other two are from
$\{c_j^1, c_j^2, c_j^3\}$.
For the one vertex in the above set which is not located by
the vertices in $S$ in the clause gadget,
is located by its neighborhood in $V$. 
Finally, any two vertices of the form $c_j^{r}$ and $c_{j'}^{r'}$ 
that are not located by corresponding clause vertices in $S$
are located from one another by the fact they are 
associated with two different variables or two different appearances
of the same variable, and hence by construction,
their neighborhood in $V$ is different.

\item Third, let us consider pairs of vertices not belonging to $S$ and where one vertex belongs to a clause gadget and the other belongs to a variable gadget. To that end, we only need to consider pairs of vertices which are of distance at most~2 from each other, that is, pairs of the form $(c_j^r, x_i^\ell)$ or $(c_j^r, \neg x_i^\ell)$, where $i \in [n]$, $j \in [m]$, $r \in [3]$ and $\ell \in [2]$. Without loss of geenrality, let us consider the pair $(c_j^r, x_i^\ell)$, where $c_j^r, x_i^\ell \notin S$. This implies that the literal $x_i$ does not appear at the $\ell^{th}$ time at the $r^{th}$ position of the clause $C_j$, or else, by the assumption $c_j^r \notin S$, we would have $\pi(x_i) = \true$ making $x_i^\ell \in S$, a contradiction to our assumption. Hence, by construction, the vertices $c_j^r$ and $x_i^\ell$ have different neighborhoods in $V$ and thus, are located by $S$. A similar argument for $(c_j^r, \neg x_i^\ell)$ proves that the latter pair is also located by $S$.  
 
\end{itemize}

This proves that $S$ is a locating set of $G$.
\end{proof}

\begin{lemma}
\label{lem:backward-3,3-sat-to-LD}
If $(G, k)$ is a \yes-instance of \LD,
then $\psi$ is a \yes-instance of \textsc{$(3,3)$-SAT}.
\end{lemma}

\begin{proof}
Suppose $S$ is a locating-dominating set of $G$ of size $k = 4n+3m+2p$. 
Recall that every vertex $v$ in the connection portal $V$
is adjacent to a pendant vertex.
Hence, by Observation~\ref{obs:nbr-of-pendant-vertex-in-sol},
it is safe to assume that $S$ contains $V$.
This implies $|S \setminus V| = 4n + 3m$. 
Using the same observation, it is safe to assume that 
$\alpha^0_i, \beta^0_i$ and $\gamma^0_j$ are in $S$
for every $i \in [n]$ and $j \in [m]$.
As $\alpha^3_i$ is adjacent to only $\alpha^0_1$ in $S$,
set $S \setminus \{\alpha^0_i, \beta^0_i\}$ 
contains at least one vertex in the closed
the neighborhood of $\alpha^1_i$ and one from the closed
neighborhood of $\alpha^2_i$.
By the construction, these two sets are disjoint.
Hence, set $S$ contains at least four vertices from 
variable gadget corresponding to every variable $x_i$ in $X$.
Using similar arguments, $S$ has at least three vertices
from clause gadget corresponding to $C_j$
for every clause in $C$.
The cardinality constraints mentioned above,
implies that 
$S$ contains exactly four vertices from each
variable gadget and exactly three vertices from
each clause gadget.
Using this observation, we prove the following two claims.

\begin{claim}
We may assume that, for each variable gadget corresponding to variable $x_i$,
$S$ contains either $\{\alpha^0_i, \beta^0_i, x^1_i, x^2_i\}$
or $\{\alpha^0_i, \beta^0_i, \neg x^1_i, \neg x^2_i\}$.
\end{claim}
\begin{claimproof}
Note that $S$ is a locating-dominating set of $G$ such that 
$|S \cap A_i| = 4$.
 Let $X_i = \{x_i^1, x_i^2\}$ and $\neg X_i = \{\neg {x}_i^1, 
\neg {x}_i^2\}$.
We prove that $S$ contains exactly one of $X_i$ and $\neg X_i$.
First, we show that $|(S \cap A_i) \setminus 
\{\alpha_i^0,\alpha_i^3,\beta_i^0,\beta_i^3\}| = 2$. 
Define 
$R_i^1 = \{\alpha_i^1,x_i^1,\neg {x}_i^1\}$, 
$R_i^2 = \{\alpha_i^2,x_i^2,\neg {x}_i^2\}$,
$G_i^1 = \{\beta_i^1,\neg {x}_i^1,x_i^2\}$, and
$G_i^2 = \{\beta_i^2, \neg {x}_i^2,x_i^1\}$.
Note that sets $R_i^1$ and $R_i^2$ (respectively, $G_i^1$ and
$G_i^2$) are disjoint.
Also, $(R_i^1 \cap G_i^2) \cup (R_i^2 \cap G_i^1) = \{x_i^1, x_i^2\}$
and $(R_i^1 \cap G_i^1) \cup (R_i^2 \cap G_i^2) = \{\neg x_i^1, \neg x_i^2\}$.
To distinguish the vertices in pairs $(\alpha_i^1,\alpha_i^2)$, 
$(\alpha_i^1,\alpha_i^3)$ and $(\alpha_i^2,\alpha_i^3)$, the set 
$S$ must contain at least one vertex from each
$R_i^1$ and $R_i^2$
or at least one vertex from each $G_i^1$ and $G_i^2$.
Similarly, to distinguish the vertices in pairs 
$(\beta_i^1,\beta_i^2)$, 
$(\beta_i^1,\beta_i^3)$ and $(\beta_i^2,\beta_i^3)$, the set 
$S$ must contain at least one vertex from each
$R_i^1$ and $R_i^2$
or at least one vertex from each $G_i^1$ and $G_i^2$.
As both of these conditions need to hold simultaneously,
$S$ contains either $X_i$ or $\neg X_i$.
\end{claimproof}

\begin{claim}
We may assume that, 
for each clause gadget corresponding to clause $C_j$,
$S$ contains either $\{\gamma^0_i, c^2_i, c^3_i \}$
or $\{\gamma^0_i, c^1_i, c^3_i \}$ or
$\{\gamma^0_i, c^1_i, c^2_i \}$.
\end{claim}

\begin{claimproof}
Note that $S$ is a dominating set that contains $3$ vertices 
corresponding to the clause gadget corresponding to every clause in $C$.
Moreover, $c^1_j, c^2_j, c^3_j$ separate the
vertices in the remaining clause gadget from the rest of the graph. 
Also, as mentioned before, without loss of generality, $C$ contains
$\gamma_j^0$.
We define $C_j^\star = \{c_j^1, c_j^2, c_j^3\}$ and
prove that
$S$ contains exactly two vertices from this set.
Assume that $S \cap C_j^\star = \emptyset$,
then $S$ needs to include all the three vertices $\gamma^1_j, \gamma^2_j, \gamma^3_j$
to locate every pair of vertices in the clause gadget.
This, however, contradicts the cardinality constraints.
Assume that  $|S \cap C_j^\star| = 1$ and
without loss of generality, suppose $S \cap C_j^\star = 
\{c_j^2\}$. 
If $S \cap \{\gamma_j^{1}, \gamma_j^{2}, \gamma_j^{3}\} = \gamma_j^{2}$, 
then the pair $(\gamma_j^{1}, \gamma_j^{3})$ is not located by $S$, a 
contradiction. 
Now suppose $S \cap \{\gamma_j^{1}, \gamma_j^{2}, \gamma_j^{3}\} = \gamma_j^{1}$,
then the pair $(\gamma_j^{2}, \gamma_j^{4})$ is not located by $S$, again
a contradiction. 
The similar argument holds for other cases.
As both cases, this leads to contradictions, we have 
$|S \cap C_j^\star| = 2$.
\end{claimproof}

Using these properties of $S$, we present a way to construct 
an assignment $\pi$ for $\psi$.
If $S$ contains $\{\alpha^0_i, \beta^0_i, x^1_i, x^2_i\}$,
then set $\pi(x_i) = \true$, otherwise set $\pi(x_i) = \false$.
The first claim ensures that this is a valid assignment.
We now prove that this assignment is also a satisfying assignment.
The second claim implies that for any clause gadget corresponding
to clause $C_j$, there is exactly one vertex \emph{not} adjacent to 
any vertex in $S \cap B_j$.
Suppose $c_j^{r}$ is such a vertex and it is the $\ell^{th}$ positive
appearance of variable $x_i$.
Since $x_i^{\ell}$ and $c_j^{r}$ have identical neighborhood in $V$,
and $S$ is a locating-dominating set in $G$,
$S$ contains $x_i^{\ell}$, and hence $\pi(x_i) = \true$.
A similar reason holds when $x_i$ appears negatively.
Hence, every vertex in the clause gadget that is not dominated by vertices
of $S$ in the clause gadget corresponds to the variable
set to \true\, and this makes the clause satisfied.
This implies $\pi$ is a satisfying assignment of $\psi$
which concludes the proof of the lemma.
\end{proof}

We are now in a position to prove the conditional lower bounds about \LD\
mentioned in Theorem~\ref{thm:TW}: \LD does not admit an algorithm 
running in time  $2^{2^{o(\tw)}} \cdot {\rm poly}(n)$.


\begin{proof}[Proof of the \LD part of Theorem~\ref{thm:TW}]
Assume that there is an algorithm $\calA$ that, given an instance 
$(G, k)$ of \LD, runs in time $2^{2^{o(\tw)}} \cdot n^{\calO(1)}$
and correctly determines whether it is \yes-instance.
Consider the following algorithm that takes as input
an instance $\psi$ of \textsc{$(3, 3)$-SAT} and determines
whether it is a \yes-instance.
It first constructs an equivalent instance $(G, k)$ of \LD\
as mentioned in this subsection.
Then, it calls algorithm $\calA$ as a subroutine and returns 
the same answer.
The correctness of this algorithm follows from the correctness
of algorithm $\calA$, Lemma~\ref{lem:forward-3,3-sat-to-LD} 
and Lemma~\ref{lem:backward-3,3-sat-to-LD}.
Note that since each component of $G - V$ is of constant order,
the vertex integrity (and thus treewidth) of $G$ is $\calO(|V|)$.
By the asymptotic estimation of the central binomial coefficient, 
$\binom{2p}{p}\sim \frac{4^p}{\sqrt{\pi \cdot p}}$~\cite{Sperner}. 
To get the upper bound of $2p$, we scale down the asymptotic function and have that $4n \leq \frac{4^p}{2^p}=2^p$.
As we choose the smallest possible value of $p$ such that $2^p \geq 4n$, we can choose $p = \log n + 3$. 
Therefore, $p=\calO(\log (n))$. 
And hence, $|V| = \calO(\log (n))$ which implies $\tw(G) = \calO (\log n)$.
As the other steps, including the construction of the instance 
of \LD, can be performed in the polynomial step,
the running time of the algorithm for \textsc{$(3, 3)$-SAT}
is $2^{2^{o(\log (n))}} \cdot n^{\calO(1)} = 2^{o(n)} \cdot n^{\calO(1)}$.
This, however, contradicts Proposition~\ref{prop:3-3-SAT-ETH-lb}.
Hence, our assumption is wrong, and
\LD\ does not admit an algorithm running in time
$2^{2^{o(tw)}} \cdot |V(G)|^{\calO(1)}$, 
unless the {\ETH} fails.
\end{proof}

\section{\LD Parameterized by the Solution Size}
\label{sec:LD-sol-size}

In this section, we study the parameterized complexity of 
\LD when parameterized 
by the solution size $k$.
In the first subsection, we formally prove that the problem admits a kernel
with $\calO(2^{k})$ vertices, and hence a simple \FPT\
algorithm running in time $2^{\calO(k^2)}$.
In the second subsection, we prove that both results mentioned above are optimal
under the \ETH.

\subsection{Upper bound}

\begin{proposition}\label{prop:LD-algo}
\LD admits a kernel with $\calO(2^k)$ vertices and 
an algorithm running in time $2^{\calO(k^2)}+\calO(k\log n)$.
\end{proposition}
\begin{proof}
Slater proved that for any graph $G$ on $n$ vertices with a locating-dominating set of size $k$, we have $n\leq 2^k+k-1$~\cite{slater1988dominating}.
Hence, if $n>2^k+k-1$, we can return a trivial \no\ instance (this
check takes time $\calO(k\log n)$). 
Otherwise, we have a kernel with $\calO(2^k)$ vertices.
In this case, we can enumerate all
subsets of vertices of size $k$, and for each of them, check in
quadratic time if it is a valid solution. Overall, this takes time
${n\choose k}n^2$; since $n\leq 2^k+k-1$, this is
${2^{\calO(k)}\choose k} \cdot 2^{\calO(k)}$, which is $2^{\calO(k^2)}$.
\end{proof}

\subsection{Lower Bound}
\label{subsec:LD-lower-bound-solution-size}

In this section, we prove Theorem~\ref{thm:locating-dom-set-sol-size-lb} which
states that both the results mentioned in the previous subsection are optimal
under the \ETH. We restate it for the reader's convenience.

\locdomsetsolution*

\subparagraph*{Reduction.}

To prove the theorem,  we present a reduction that takes as input 
an instance $\psi$, with $n$ variables, of \textsc{3-SAT} and returns an instance 
$(G, k)$ of \LD such that $|V(G)| = 2^{\calO(\sqrt{n})}$ and $k = \calO(\sqrt{n})$.
By adding dummy variables in each set, we can assume that $\sqrt{n}$ is an even integer.
Suppose the variables are named $x_{i, j}$ for $i, j \in [\sqrt{n}]$.

We will make use of \emph{bit-representation gadgets}, as formalized in~\cite{FoucaudGK0IST24}, which are useful to ensure that certain sets of vertices are located from each other and also from the rest of the graph. For a set $X$ to be located, a corresponding bit-representation gadget is of logarithmic size (in terms of the size of $X$) and assign a distinct set of neighbors of the gadget to each vertex of $X$. By attaching pendant vertices to each vertex in the bit-representation gadget, we ensure (using Observation~\ref{obs:nbr-of-pendant-vertex-in-sol}) that they can be assumed to be part of any solution and hence, locate $X$.

The reduction constructs graph $G$ as follows:
\begin{itemize} 
\item 
It partitions the variables of $\psi$ into $\sqrt{n}$ many \emph{buckets}
$X_1,  X_2,  \dots ,X_{\sqrt{n}}$ such that each bucket contains exactly $\sqrt{n}$ many variables.
Let $X_i = \{x_{i, j}\ |\ j \in [\sqrt{n}]\}$ for all $i\in [\sqrt{n}]$.
\begin{itemize}
\item For every $X_i$, it constructs set $A_i$ of $2^{\sqrt{n}}$ new vertices, 
$A_i=\{a_{i,  \ell}\mid \ell \in [2^{\sqrt{n}}]\}$.
Each vertex in $A_i$ corresponds to a unique assignment of variables in $X_i$.
Let $A$ be the collection of all the vertices added in this step.

\item For every $X_i$, the reduction adds a path on three vertices $b^{\circ}_i$, $b'_i$, and $b^{\star}_i$ 
with edges $(b^{\circ}_i, b'_i)$ and $(b'_i, b^{\star}_i)$. 
Suppose $B$ is the collection of all the vertices added in this step.
\item For every $X_i$, it makes $b^{\circ}_i$ adjacent to every vertex in $A_i$.
\end{itemize}

\item 
For every clause $C_j$, the reduction adds a pair of vertices $c^{\circ}_j, c^{\star}_j$.
For a vertex $a_{i, \ell} \in A_i$ for some $i \in [\sqrt{n}]$, and $\ell \in [2^{\sqrt{n}}]$,
if the assignment corresponding to vertex $a_{i, \ell}$ satisfies clause $C_j$, 
then it adds edge $(a_{i, \ell}, c^{\circ}_j)$.

\item The reduction adds a bit-representation gadget to locate set $A$.
Once again, informally speaking, it adds some supplementary vertices such 
that it is safe to assume that these vertices are present in a locating-dominating set, and they locate every vertex in $A$. More precisely:
\begin{itemize}
\item First, set $q := \lceil \log(|A|) \rceil+1$.
This value for $q$ allows to uniquely represent each integer in $[|A|]$ by 
its bit-representation in binary (starting from $1$ and not $0$). 
\item For every $i \in [q]$, the reduction adds two vertices $y_{i, 1}$ and $y_{i, 2}$ and edge $(y_{i, 1},  y_{i, 2})$.
\item For every integer $q' \in [|A|]$, let $\bit(q')$ denote the binary representation of $q'$ using $q$ bits.
Connect $a_{i, \ell} \in A$ with $y_{i, 1}$ if the $i^{th}$ bit in $\bit((i - 1) \cdot 2^{\sqrt{n}} + \ell)$ is $1$.
\item It adds two vertices $y_{0, 1}$ and $y_{0, 2}$, and edge $(y_{0, 1}, y_{0, 2})$.
It also makes every vertex in $A$ adjacent to $y_{0, 1}$.

Let $\bitrep(A)$ be the collection of vertices $y_{i,1}$ added in this step, and let $Y$ be the collection of vertices $y_{i,1}$ and $y_{i,2}$ added in this step (over all $i \in \{0\} \cup [q]$).
\end{itemize}

\item Finally, the reduction adds a bit-representation gadget to locate set $C$.
However, it adds the vertices in such a way that for any pair $c^{\circ}_j, c^{\star}_j$, the supplementary vertices
adjacent to them are identical.
\begin{itemize}
\item The reduction sets $p := \lceil \log(|C|/2) \rceil+1$ and
for every $i \in [p]$, it adds two vertices $z_{i, 1}$ and $z_{i, 2}$ and edge $(z_{i, 1},  z_{i, 2})$.
\item For every integer $j \in [|C|/2]$, let $\bit(j)$ denote the binary representation of $j$ using $p$ bits.
Connect $c^{\circ}_{j}, c^{\star}_j \in C$ with $z_{i, 1}$ if the $i^{th}$ bit in $\bit(j)$ is $1$.
\item It adds two vertices $z_{0, 1}$ and $z_{0, 2}$, and edge $(z_{0, 1}, z_{0, 2})$.
It also makes every vertex in $C$ adjacent to $y_{0, 1}$.

Let $\bitrep(C)$ be the collection of the vertices $z_{i,1}$ 
added in this step, and let $Z$ be the collection of vertices $z_{i,1}$ and $z_{i,2}$ added in this step (over all $i \in \{0\} \cup [q]$).
\end{itemize}
\end{itemize}

This completes the construction; see Figure~\ref{fig:locating-dom-set-sol-size} for an illustration.

\begin{figure}[t]
\centering
\includegraphics[scale=0.75]{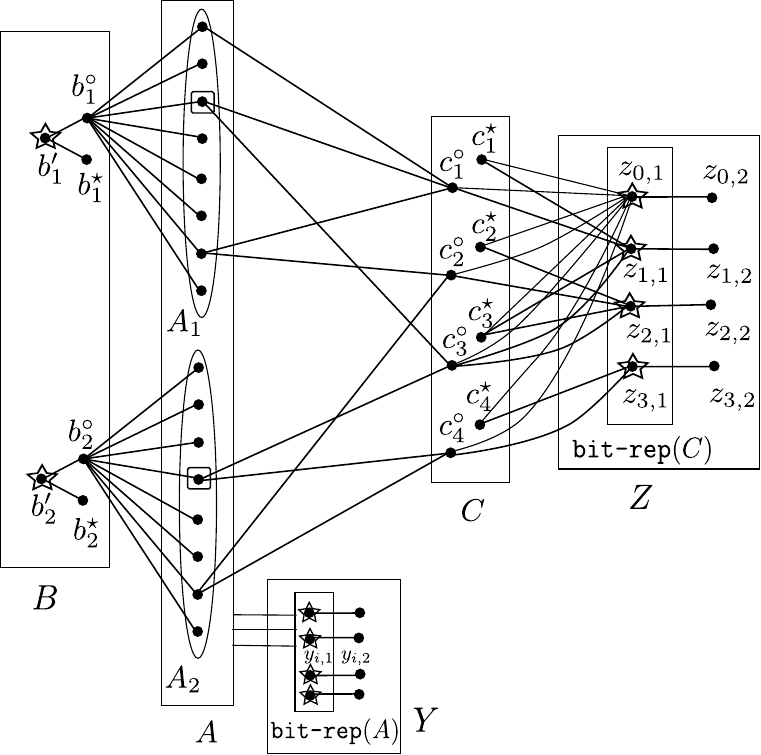}
\caption{An illustrative example of the graph constructed by the reduction used in Section~\ref{subsec:LD-lower-bound-solution-size}.
Suppose an instance $\psi$ of \textsc{3-SAT} has $n = 9$ variables and $4$ clauses.
We do not show the third variable bucket and explicit edges across $A$ and $\bitrep(A)$ for brevity. Vertices with a star boundary are those that we can assume to be
in any locating-dominating set, without loss of generality.
The square boundaries correspond to the selection of other vertices in the solution.
\label{fig:locating-dom-set-sol-size}}
\end{figure}

The reduction sets 
$$k = |B|/3 + (\lceil \log(|A|) \rceil + 1 + 1) + \lceil (\log(|C|/2) \rceil + 1 + 1) + \sqrt{n} = \calO(\sqrt{n})$$
as $|B| = 3\sqrt{n}$, $|A| = \sqrt{n} \cdot 2^{\sqrt{n}}$, and $|C| = \calO(n^3)$,
and returns $(G, k)$ as a reduced instance. 



The vertices of $G$ are naturally partitioned into the five sets
$A$, $B$, $C$, $Y$ and $Z$ defined above. 
Furthermore, we partition $B$ into $B'$, $B^{\circ}$ and $B^{\star}$
as follows: $B' = \{b'_i \mid i \in [\sqrt{n}]\}$,
$B^{\circ} = \{b^{\circ}_i \mid i \in [\sqrt{n}]\}$, and
$B^{\star} = \{b^{\star}_i \mid i \in [\sqrt{n}]\}$.
Define $Y_1$, $Y_2$, $Z_1$ and $Z_2$ in a similar way.
Note that $B^{\star}$, $Y_2$, and $Z_2$ together contain
all pendant vertices.

\begin{lemma}
\label{lemma:locating-dom-set-sol-size-forward}
If $\psi$ is a \yes-instance of \textsc{3-SAT}, then 
$(G, k)$ is a \yes-instance of \LD.
\end{lemma}
\begin{proof}
  Suppose $\pi: X \mapsto \{\true, \false\}$ is a satisfying assignment for $\psi$.
Using this assignment, we construct a locating-dominating set $S$ of $G$ of size at most $k$.
Initialise set $S$ by adding all the vertices in $B' \cup Y_1 \cup Z_1$.
At this point, the cardinality of $S$ is $k - \sqrt{n}$.
We add the remaining $\sqrt{n}$ vertices as follows:
Partition $X$ into $\sqrt{n}$ parts $X_1,\dots, X_i,  \dots, X_{\sqrt{n}}$ as specified in the reduction, and
define $\pi_i$ for every $i \in [\sqrt{n}]$ by restricting the assignment $\pi$
to the variables in $X_i$.
By the construction of $G$, there is a vertex $a_{i, \ell}$ in $A$ corresponding to 
the assignment $\pi_i$.
Add that vertex to $S$.
It is easy to verify that the size of $S$ is at most $k$.
We next argue that $S$ is a locating-dominating set.

The vertices in $B$ are located from all other vertices by the
vertices in $B'$. Moreover, pairs of the form
$\{\{b^{\circ}_i,\{b^{\star}_i\}$ are located by the vertices in
$A\cap S$.

Consider set $A$. 
By the property of a bit-representation gadget,
every vertex in $A$ is adjacent to a unique set of vertices 
in $\bitrep(A)$. 
Consider a vertex $a_{i,\ell}$ in $A$ such that the bit-representation
of $((i - 1) \cdot 2^{\sqrt{n}} + \ell)$ contains a single $1$ at the $j^{th}$ location.
Hence, both $y_{j, 2}$ and $a_{i,\ell}$ are adjacent to the same 
vertex, viz $y_{j, 1}$ in  $\bitrep(A) \setminus \{y_{0, 1}\}$.
However, this pair of vertices is located by $y_{0, 1}$ which 
is adjacent to $a_{i,\ell}$ and \emph{not} to $y_{j, 2}$.
Also, as the bit-representation of vertices starts from $1$,
there is no vertex in $A$ which is adjacent to only $y_{0, 1}$ in 
$\bitrep(A)$. Hence, $\bitrep(A)$ locates all pairs of vertices in $A\cup Y$.

Using similar arguments for $C$, $\bitrep(C)$ and $Z$ and the properties of 
bit-representation gadgets, we can conclude that 
$\bitrep(A) \cup \bitrep(C)$ locates all pairs of vertices of $G$, 
apart from the pairs of the form $(b^{\circ}_j, b^{\star}_j)$
and $(c^{\circ}_j, c^{\star}_j)$. 

\begin{table}[t]
\centering
\begin{tabular}{|l|c|c|c|c|c|c|c|c|c|c|}
\hline
Set & $B'$ & $B^{\circ}$ & $B^{\star}$ & $A$ & $Y_1$ & $Y_2$ & $C^{\circ}$ & $C^{\star}$ & $Z_1$ & $Z_2$ \\  
\hline
Dominated by & $B'$ & $B'$ & $B'$ & $Y_1$ & $Y_1$ & $Y_1$ & $Z_1$ & $Z_1$ & $Z_1$ & $Z_1$ \\  
\hline
Located by & - & $B' + A$ & $B'$ & $Y_1$ & - & $Y_1$ & $Z_1 + A$ & $Z_1$ & - & $Z_1$ \\  
\hline
\end{tabular}
\caption{Partition of $V(G)$ and the corresponding set that dominates and locates the vertices in each part.
\label{table:loc-dom-set-sol-size} }
\end{table}

By the construction, the sets mentioned
in the second row of Table~\ref{table:loc-dom-set-sol-size}
dominate the vertices mentioned in the sets in the respective first rows.
Hence, $S$ is a dominating set.
We need to prove that the location of only those vertices of a set which are dominated by other vertices of the same set.
First, consider the vertices in $B^{\circ} \cup B^{\star}$.
Recall that every vertex of the form $b^{\circ}_i$ and $b^{\star}_i$ is
adjacent to $b'_i$ for every $i \in [\sqrt{n}]$.
Hence, the only pair of vertices that needs to be located are of the form
$b^{\circ}_i$ and $b^{\star}_i$.
However, as $S$ contains at least one vertex from $A_i$, a vertex in $S$ is adjacent to $b^{\circ}_i$ and not adjacent to $b^{\star}_i$.
Now, consider the vertices in $A$ and $Y_2$.
Note that every vertex in $Y_2$ is adjacent to precisely one vertex in $Y_1$. However, every vertex in $A$ is adjacent to at least two vertices in $Y_1$ (one of which is $y_{0, 1}$).
Hence, every vertex in $A \cup Y_2$ is located.
Using similar arguments, every vertex in $C \cup Z_2$ is also located.

The only thing that remains to argue is that every pair of vertices $c^{\circ}_j$ and $c^{\star}_j$ is located.
As $\pi$ is a satisfying assignment, at least one of its restrictions, say $\pi_i$, is a satisfying
assignment for clause $C_j$.
By the construction of the graph, the vertex corresponding to $\pi_i$ is adjacent to $c^{\circ}_j$
but not adjacent to $c^{\star}_j$.
Also, such a vertex is present by the construction of $S$.
Hence, there is a vertex in $S \cap A_i$ that locates $c^{\circ}_j$ from $c^{\star}_j$.
This concludes the proof that $S$ is a locating-dominating set of $G$ of size $k$.
Hence, if $\psi$ is a \yes-instance of \textsc{3-SAT}, then $(G, k)$ is a \yes-instance
of \LD.
\end{proof}

\begin{lemma}
\label{lemma:locating-dom-set-sol-size-backward}
If $(G, k)$ is a \yes-instance of \LD, then $\psi$ is a \yes-instance of \textsc{3-SAT}
\end{lemma}
\begin{proof}
Suppose $S$ is a locating-dominating set of $G$ of size at most $k$.
We construct a satisfying assignment $\pi$ for $\psi$.
Recall that $B^{\star}_i$, $Y_2$, and $Z_2$ contain exactly all pendant vertices of $G$.
By Observation~\ref{obs:nbr-of-pendant-vertex-in-sol}, it is safe to assume that
every vertex in $B'$, $Y_1$, and $Z_1$ is present in $S$.

Consider the vertices in $B^{\circ}$ and $B^{\star}$.
As mentioned before, every vertex of the form $b^{\circ}_i$ and $b^{\star}_i$ 
is adjacent to $b'_i$, which is in $S$.
By the construction of $G$, only the vertices in $A_i$ are adjacent to $b^{\circ}_i$
but not adjacent to $b^{\star}_i$.
Hence, $S$ contains at least one vertex in $A_i \cup \{b^{\circ}_i, b^{\star}_i\}$.
As the number of vertices in $S \setminus (B' \cup Y_1 \cup Y_2)$ is at most $\sqrt{n}$, 
$S$ contains exactly one vertex from $A_i \cup \{b^{\circ}_i, b^{\star}_i\}$ for every $i \in [\sqrt{n}]$.
Suppose $S$ contains a vertex from $\{b^{\circ}_i, b^{\star}_i\}$.
As the only purpose of this vertex is to locate a vertex in this set, 
it is safe to replace this vertex with any vertex in $A_i$. 
Hence, we can assume that $S$ contains exactly one vertex in $A_i$ for every $i \in [\sqrt{n}]$.

For every $i \in [\sqrt{n}]$, let $\pi_i: X_i \mapsto \{\true, \false\}$
be the assignment of the variables in $X_i$ corresponding to the vertex of $S\cap A_i$.
We construct an assignment $\pi: X \mapsto \{\true, \false\}$ 
such that that $\pi$ restricted to $X_i$ is identical to $\pi_i$.
As $X_i$ is a partition of variables in $X$, and every vertex in $A_i$
corresponds to a valid assignment of variables in $X_i$,
it is easy to see that $\pi$ is a valid assignment.
It remains to argue that $\pi$ is a satisfying assignment.
Consider a pair of vertices $c^{\circ}_j$ and $c^{\star}_j$ 
corresponding to clause $C_j$.
By the construction of $G$, both these vertices have identical 
neighbors in $Z_1$, which is contained in $S$.
The only vertices that are adjacent to $c^{\circ}_j$ and
are not adjacent to $c^{\star}_j$ are in $A_i$ for some $i\in [\sqrt{n}]$
and correspond to some assignment that satisfies clause $C_j$.
As $S$ is a locating-dominating set of $G$, there is at least one vertex
in $S \cap A_i$, that locates  $c^{\circ}_j$ from $c^{\star}_j$.
Alternately, there is at least one vertex in $S \cap A_i$ that corresponds 
to an assignment that satisfies clause $C_j$.
This implies that if $(G, k)$ is a \yes-instance of \LD,
then $\phi$ is a \yes-instance of \textsc{3-SAT}.
\end{proof}

We are now in a position to prove Theorem~\ref{thm:locating-dom-set-sol-size-lb} 
which states that unless the \ETH\ fails, \LD 
does not admit an algorithm running in time $2^{o(k^2)} \cdot n^{\calO(1)}$,
nor a polynomial-time kernelization algorithm that reduces the solution size 
and outputs a kernel with $2^{o(k)}$ vertices.

\begin{proof}[Proof of Theorem~\ref{thm:locating-dom-set-sol-size-lb}]
Assume that there exists an algorithm, say $\calA$, that takes as input an instance $(G, k)$
of \LD and correctly concludes whether it is a \yes-instance
in time $2^{o(k^2)}\cdot |V(G)|^{\calO(1)}$.
Consider algorithm $\calB$ that takes as input an instance $\psi$ of \textsc{3-SAT},
uses the reduction above to get an equivalent instance $(G, k)$ of 
\LD, and then uses $\calA$ as a subroutine.
The correctness of algorithm $\calB$ follows from Lemma~\ref{lemma:locating-dom-set-sol-size-forward},
Lemma~\ref{lemma:locating-dom-set-sol-size-backward}, 
and the correctness of algorithm $\calA$.
From the description of the reduction and the fact that $k = \sqrt{n}$, 
the running time of algorithm $\calB$ is 
$2^{\calO{(\sqrt{n})}} + 2^{o(k^2)} \cdot (2^{\calO{(\sqrt{n})}})^{\calO(1)} = 2^{o(n)}$.
This, however, contradicts the \ETH.
Hence, \LD does not admit an algorithm with running time $2^{o(k^2)} \cdot |V(G)|^{\calO(1)}$
unless the \ETH\ fails.
 
For the second part of Theorem~\ref{thm:locating-dom-set-sol-size-lb}, assume that such a kernelization algorithm exists.
Consider the following algorithm for \textsc{$3$-SAT}.
Given a \textsc{$3$-SAT} formula on $n$ variables, 
it uses the above reduction to get an equivalent instance of $(G, k)$ such that 
$|V(G)| = 2^{\calO(\sqrt{n})}$ and $k = \calO(\sqrt{n})$.
Then, it uses the assumed kernelization algorithm
to construct an equivalent instance $(H, k')$ such that 
$H$ has $2^{o(k)}$ vertices and $k' \le k$.
Finally, it uses a brute-force algorithm,  running in time $|V(H)|^{\calO(k')}$,
to determine whether the reduced instance, equivalently the input
instance of \textsc{$3$-SAT},  is a \yes-instance.
The correctness of the algorithm follows from the 
correctness of the respective algorithms and our assumption.
The total running time of the algorithm is 
$2^{\calO(\sqrt{n})} + (|V(G)| + k)^{\calO(1)} + |V(H)|^{\calO(k')}
= 2^{\calO(\sqrt{n})} + (2^{\calO(\sqrt{n})})^{\calO(1)} + (2^{o(\sqrt{n})})^{\calO(\sqrt{n})}
= 2^{o(n)}$.
This, however, contradicts the \ETH.
Hence, \LD does not admit a polynomial-time kernelization algorithm that reduces 
the solution size and returns a kernel with $2^{o(k)}$ vertices
unless the \ETH\ fails.
\end{proof}

\section{Modifications for \TCPB}
\label{sec:modifications-for-test-cover}

In this section, we present the small (but crucial) modifications
required to obtain similar results as mentioned in previous 
section for \TCPB.
As in this problem, one only need to `locate' elements, 
the reductions in this case are often simpler than
the corresponding reductions for \LD\ to prove a conditional lower
bound. 
We restate the problem definition for the readers' convenience.

\defproblem{\TCPB}{A set of items $U$, a collection of subsets of $U$ called \emph{tests} and denoted by $\calF$, and an integer $k$.}{Does 
there exist a collection of at most $k$ tests
such that for each pair of items, there is a test 
that contains exactly one of the two items?}

\noindent We also recall the incidence (bipartite) graph
$G$ on $n$ vertices with bipartition $\langle R, B \rangle$
of $V(G)$ such that sets $R$ and $B$ contain a vertex
for every set in $\calF$ and for every item in $U$, respectively,
and $r \in R$ and $b \in B$ are adjacent 
if and only if the set corresponding to $r$ contains 
the element corresponding to $b$.
We note that description of the incidence (bipartite) graph $G$
is sufficient to describe the instance of \TCPB.
We find it notationally cleaner to work with encoding of the instance.

\subsection{Parameterization by Treewidth}

\subsubsection{Upper Bound}
\label{sec:test-cover-tw-algo}
In this section, we outline the dynamic programming approach for {\TCPB} and use it to prove the following theorem.

\begin{theorem}
\label{thm:TC-tw-algo}
{\TCPB} admits an algorithm running in time 
\(2^{2^{\calO(\tw)}} \cdot n^{\calO(1)}\),  
where \(\tw\) is the treewidth and \(n\) is the order of the incidence graph of the input.
\end{theorem}

Let \(G\) be a bipartite graph on \(n\) vertices with a bipartition \(V(G) = \langle R, B \rangle\).  
The objective is to decide whether there exists a set of at most \(k\) vertices from \(R\) such that, for every pair \(x, y \in B\), there exists \(r \in R\) with \(rx \in E(G)\) but \(ry \notin E(G)\).  
Without loss of generality, we assume that we are given a nice tree-decomposition \(\calT = (T, \{X_t\}_{t \in V(T)})\) of \(G\) of width at most \(2{\tw}(G) + 1\).

The overall structure of the algorithm and the motivation for defining the DP-states follows the framework in Section~\ref{sub-sec:lds-tw-algo}.  
For completeness, we recall the relevant notions, adapted to the present setting.  
As before, for each \(t \in V(T)\), let \(G_t\) denote the subgraph of \(G\) corresponding to the subtree of \(T\) rooted at \(t\).  
We define \(R_t := R \cap V(G_t)\) and \(B_t := B \cap V(G_t)\).

\smallskip
For each node \(t\), the DP-state is defined in terms of a \emph{valid tuple}, which encodes the interaction between the partial solution and the vertices in the current bag.

\begin{definition}[Valid Tuple]
Let \( t \in V(T) \) be a node in the tree-decomposition. 
A tuple \(\langle Y, W, W_0, \calY, \mathbb{W} \rangle\) is \emph{valid} at \(t\) if:
\begin{itemize}[nolistsep]
	\item \(Y \subseteq R \cap X_t\) and \(W \subseteq B \cap X_t\),  
	      with \(N(Y) \cap X_t \subseteq W\) and \(W_0 \subseteq W\);
	\item \(\calY\) is a family of subsets of \(Y\);
	\item \(\mathbb{W}\) is a collection of  
	      pairs \((w_1, w_2)\) with \(w_1, w_2 \in W\), or  
	      \((w_1, +)\) with \(w_1 \in W\).
\end{itemize}
\end{definition}

\begin{definition}[Candidate Solution]
\label{def:candidate-sol-test-cover}
Let \( \tau = \langle Y, W, W_0, \calY, \mathbb{W} \rangle \) be a valid tuple at node \(t\).  
A set \(S_t \subseteq R_t\) is a \emph{candidate solution} for \(\tau\) if:
\begin{enumerate}[nolistsep]
  \item \(S_t\) locates all vertices in \(B_t \setminus X_t\);  
        that is, for any distinct \(u, v \in B_t \setminus (S_t \cup X_t)\),  
        both \(N_{G_t}(u) \cap S_t\) and \(N_{G_t}(v) \cap S_t\) are non-empty and distinct.
  \item \(Y = S_t \cap X_t\), and \(W = N(S_t) \cap X_t\).  
        Moreover, every vertex \(w \in W \setminus W_0\) has a neighbour in \(S_t \setminus Y\).
  \item \(A \in \calY\) if and only if there exists a unique \(u \in B_t \setminus (S_t \cup X_t)\)  
        such that \(N_{G_t}(u) \cap S_t = A\).
  \item For \(\mathbb{W}\):  
        \begin{itemize}[nolistsep]
          \item \((w_1, w_2) \in \mathbb{W}\) if and only if \(w_1, w_2 \in W\) and  
                \(N_{G_t}(w_1) \cap S_t = N_{G_t}(w_2) \cap S_t\);
          \item \((w_1, +) \in \mathbb{W}\) if and only if $w_1\in W$ and there exists a unique  
                \(u \in V(G_t) \setminus (S_t \cup X_t)\) with  
                \(N_{G_t}(w_1) \cap S_t = N_{G_t}(u) \cap S_t\).
        \end{itemize}
\end{enumerate}
\end{definition}

For a valid tuple \(\tau\) at node \(t\), we define \(\dptw[t, \tau]\) as the minimum size of a candidate solution for \(\tau\); if none exists, set \(\dptw[t, \tau] := \infty\).

The subsequent case analysis and recursive algorithms (and their proof of correctness) 
follow exactly as in Section~\ref{sub-sec:lds-tw-algo},  
with the additional restriction that \(Y \subseteq R\) and \(W \subseteq B\).  
Any DP-state violating these constraints can be discarded immediately.  
For instance, if a vertex \(x \in B\) is introduced at \(t\),  
the case \(x \in Y\) is ignored.  
Since this condition can be checked locally, the remainder of the argument proceeds as before, establishing Theorem~\ref{thm:TC-tw-algo}.

\subsubsection{Lower Bound}
In this subsection, we prove the lower bound
mentioned in Theorem~\ref{thm:TW}, 
which we restate for convenience.

\begin{theorem*}[Restating part of Theorem~\ref{thm:TW}]
Unless the \ETH\ fails,
\TCPB\ does not admit an algorithm 
running in time  $2^{2^{o(\tw)}} \cdot {\rm poly}(n)$, where $\tw$ is the treewidth and $n$ is the order of the incidence graph of the input.
\end{theorem*}

The reduction is a modification
of the reduction presented in Subsection~\ref{sub-sec:locating-dom-tw-lb}.

\subparagraph*{Reduction.}
For notational convenience, instead of specifying an instance
of \TCPB, we specify the incidence graph as mentioned 
in the definition.
The reduction takes as input an instance $\psi$ of 
\textsc{$(3,3)$-SAT} with $n$ variables and 
outputs an instance $(G, \langle R, B\rangle, k)$ of \TCPB\ such that $\tw(G) = \calO(\log(n))$. 
Suppose $X = \{x_1, \dots, x_n\}$ is the collection 
of variables and $C = \{C_1, \dots, C_m\}$ is the collection
of clauses in $\psi$.
Here, we consider 
$\langle x_1, \dots, x_n\rangle$ and
$\langle C_1, \dots, C_m \rangle$ be arbitrary but fixed
orderings of variables and clauses in $\psi$.
For a particular clause, the first order specifies
the first, second, or third (if it exists) variable in the clause 
in a natural way.
The second ordering specifies the first/second positive/negative
appearance of variables in $X$ in a natural way. 

The reduction constructs a graph $G$  
after initializing $R = B = \emptyset$. We again make use of the set-representation gadget from~\cite{FoucaudGK0IST24}, as in Section~\ref{sub-sec:locating-dom-tw-lb}.
\begin{itemize}
\item 
To construct a variable gadget for $x_i$,
it adds a blue vertex $\beta_i$ to $B$,
two red vertices $\alpha^1_i, \alpha^2_i$ to $R$, and
four blue vertices  $x_i^1, \neg x_i^1, x_i^2, \neg x_i^2$ to $B$.
It then adds the edges as shown in Figure~\ref{fig:test-cover-tw-lb}.
Let $A_i$ be the collection of these seven vertices.
Define $X_i := \{x_i^1, \neg x_i^1, x_i^2, \neg x_i^2\}$.

\item To construct a clause gadget for $C_j$, the reduction
adds three blue vertices $\{\gamma_j^1, \gamma_j^2, \gamma_j^3\}$ 
to $B$,
three red vertices $\{\delta_j^1, \delta_j^2, \delta_j^3\}$ 
to $R$, and three vertices
$c^1_j, c^2_j, c^3_j$ to $B$.
It adds the edges as shown in Figure~\ref{fig:test-cover-tw-lb}. 
Let $B_j$ be the collection of these nine vertices.

\item Let $p$ be the smallest positive integer such that $4n \leq 
\binom{2p}{p}$. 
Define $\calF_p$ as the collection of subsets of $[2p]$ that contains 
exactly $p$ integers (such a collection $\calF_p$ is called a \emph{Sperner 
family}). 
Define $\setrep: \bigcup_{i=1}^n X_i \mapsto 
\calF_{p}$ 
as an injective function by arbitrarily assigning a set in $\calF_p$ to 
a vertex $x^{\ell}_i$ or $\neg x^{\ell}_i$, for every $i \in [n]$ and 
$\ell \in [2]$. 
In other words, every appearance 
of a literal is assigned a distinct subset in $\calF_p$.
\item The reduction adds a \emph{connection portal} $V$, which 
is an independent set on $2p$ red vertices $v_1, v_2,  \dots, v_{2p}$, to $R$. 
For every vertex $v_q$ in $V$, 
the reduction adds a pendant blue vertex $u_q$ adjacent to
the red vertex $v_q$.
\item For every vertex $x_i^{\ell} \in X$ where $i \in [n]$ and $\ell \in [2]$, 
the reduction adds edges $(x^{\ell}_i, v_q)$ for every 
$q \in \setrep(x^{\ell}_i)$.
Similarly, it adds edges $(\neg x^{\ell}_i, v_q)$ for every 
$q \in \setrep(\neg x^{\ell}_i)$.
\item For a clause $C_j$, 
suppose variable $x_i$ appears positively for the $\ell^{th}$ time 
as the $r^{th}$ variable in $C_j$.
For example, $x_i$ appears positively for the second time
as the third variable in $C_j$.
Then, the reduction adds edges across $B$ and $V$ 
such that the vertices $c_j^r$ and $x^{\ell}_i$ have the 
same neighborhood in $V$, namely, the set 
$\{v_q : q \in \setrep(x^{\ell}_i) \}$.
Similarly, it adds edges for the negative appearance of the variables.
\item The reduction adds an isolated blue vertex $x_0$.
\end{itemize}

\begin{figure}[t]
\centering
\includegraphics[scale=0.65]{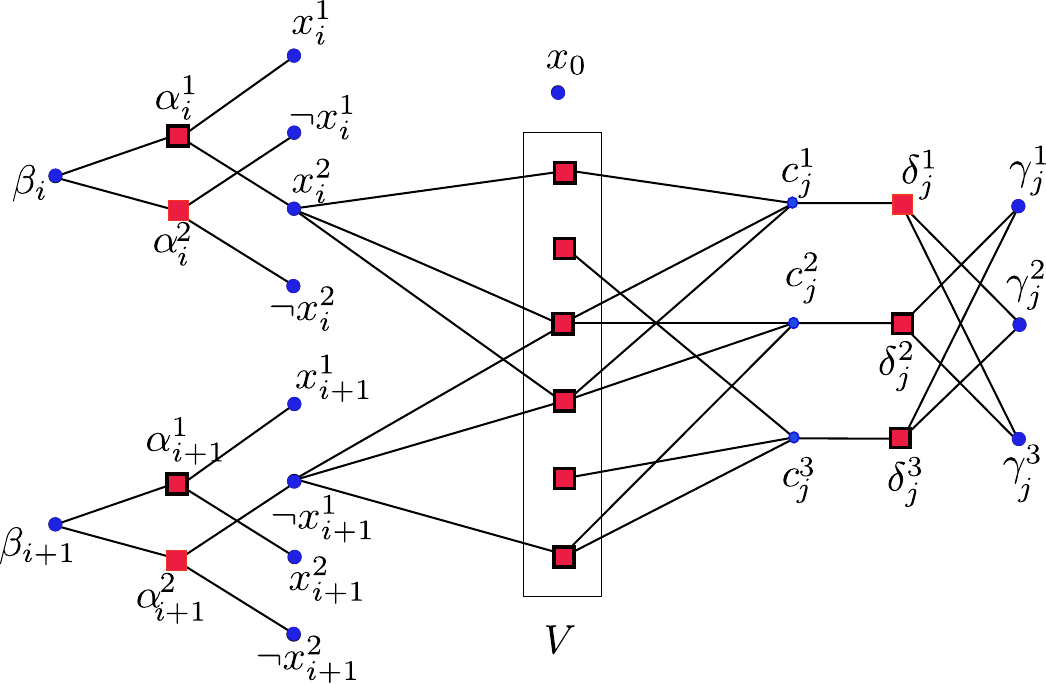}
\caption{
An illustrative example of the graph constructed by the reduction.
Red (squared) nodes denote the tests, whereas blue (filled circle) nodes
the elements.
For clarity, we do not explicitly show the pendant blue vertices adjacent to vertices in $V$.
The variable and clause gadgets are on the left and right sides
of $V$, respectively.
This is an adaptation of the same example mentioned in 
Subsection~\ref{sub-sec:locating-dom-tw-lb}.
Red vertices with a thick boundary are part of
a solution.
\label{fig:test-cover-tw-lb}}
\end{figure}

This concludes the construction of $G$. 
The reduction sets $k = n + 2m + 2p$ and 
returns $(G, \langle R, B \rangle, k)$ as the reduced instance of \TCPB.

We now prove the correctness of the reduction in the following lemma.

\begin{lemma}
\label{lem:3,3-sat-to-TCPB}
$\psi$ is a \yes-instance of \textsc{$(3,3)$-SAT} if and only if $(G, \langle R, B \rangle, k)$ is a \yes-instance of \TCPB.
\end{lemma}

\begin{proof}
Let us first assume that  $\psi$ be a \yes-instance of \textsc{$(3,3)$-SAT}. Moreover, suppose $\pi: X \mapsto \{\true, \false\}$ is a satisfying assignment
of $\psi$.
We construct a vertex subset $S$ of $G$ from the 
satisfying assignment on $\psi$ in the following manner:
Initialize $S$ by adding all the vertices in $V$.
For variable $x_i$, if $\pi(x_i) = \true$,
then include $\alpha_i^1$ in $S$,
otherwise include $\alpha_i^2$ in $S$.
For any clause $C_j$, 
if its first variable is set to \true\ then include 
$\{\delta^2_j, \delta^3_j\}$ in $S$,
if its second variable is set to \true\ then include 
$\{\delta^1_j, \delta^3_j\}$ in $S$, otherwise
include $\{\delta^1_j, \delta^2_j\}$ in $S$.
If more than one variable of a clause $C_j$ is set to \true,
we include the vertices corresponding to the smallest indexed variable
set to \true.
This concludes the construction of $S$.

It is easy to verify that $|S| = n+2m+2p = k$. To prove that $(G, \langle R, B \rangle, k)$ is a \yes-instance of \TCPB, we show that, apart from the vertex $x_0$ which has an empty neighborhood in $S$, every other blue vertex of $G$ has a (non-empty) unique neighborhood in $S$. To begin with, we see that the vertex $\beta_i$ for each $i \in [n]$ has a unique neighbor in $S$, namely either $\alpha_i^1$ or $\alpha_i^2$. Similarly, each pendant blue vertex $u_q$, for $q \in [2p]$, has a unique neighbor in $S$, namely, the vertex $v_q$. Since $V \subset S$, by the construction of $G$, the vertices of $\bigcup_{i=1}^n X_i$ have pairwise distinct neighborhoods of size $p$ in $S$. Similarly, again by construction, the vertices of $\bigcup_{i=1}^n \{c_j^1, c_j^2, c_j^3\}$ have pairwise distinct neighborhoods of size $p$ in $S$. On the other hand, since two vertices of $\{\delta_j^1, \delta_j^2, \delta_j^3\}$ belong to $S$ for each $j \in [m]$, it implies that the vertices of $\bigcup_{j=1}^m \{\gamma_j^1, \gamma_j^2, \gamma_j^3\}$ have pairwise distinct neighborhoods in $S$. Now, looking at a pair of the form $x_i^\ell$ (or equivalently $\neg x_i^\ell$) and $c_j^r$, the pair has distinct neighborhoods in $V \subset S$ if $c_j^r$ is \emph{not} the $\ell^{th}$ occurance of the literal $x_i$. Conversely, if $c_j^r$ is the $\ell^{th}$ occurance of $x_i$, then $\delta_j^r \notin S$ implies that $\pi(x_i) = \true$ which forces $\alpha_i^1 \in S$. In other words, the pair $x_i^\ell$ and $c_j^r$ have distinct neighborhoods in $S$. Finally, pairs of the form $\beta_i, x_i^\ell$ and $c_j^r, \gamma_j^s$ have pairwise distinct neighborhoods in $S$ since the vertices $x_i^\ell$ and $c_j^r$ have neighbors in $S$ and $\beta_i$ and $\gamma_j^s$ do not. This proves the forward direction of the claim.

\smallskip

We now assume that $(G, \langle R, B \rangle, k)$ is a \yes-instance of \TCPB and that $S$ is a test cover of $G$ of size $k = n+2m+2p$. Since $x_0$ is undominated by $S$, every other blue vertex must be dominated by $S$. Hence, we can assume that $V \subset S$ since each $v_q$ is adjacent to a pendant vertex $u_q$, where $q \in [2p]$, which must be covered. Similarly, for each $i \in [n]$, either $\alpha_i^1$ or $\alpha_i^2$ belongs to $S$ in order for $S$ to cover $\beta_i$. Moreover, for every $j \in [m]$, at least two vertices of $\{\delta_j^1, \delta_j^2, \delta_j^3\}$ must be in $S$ in order for the blue vertices in $\{\gamma_j^1, \gamma_j^2, \gamma_j^3\}$ to have pairwise distinct neighborhoods in $S$. This implies that $|\{\alpha_i^1, \alpha_i^2\} \cap S| = 1$ for all $i \in [n]$. Thus, setting $\pi(x_i) = \true$ if $\alpha_i \in S$ and $\pi(x_i) = \false$ otherwise, ensures a valid truth assignment for $\psi$. Moreover, we have $|\{\delta_j^1, \delta_j^2, \delta_j^3\} \cap S| = 2$ for all $j \in [m]$. Thus, if $\delta_j^r \notin S$ for some $j \in [m]$ and $r \in [3]$, it implies that $\alpha_i^1 \in S$ (respectively, $\alpha_i^2 \in S$), where $c_j^r$ is the $\ell^{th}$ occurence of $x_i$ (respectively, $\neg x_i$). This further implies that $\pi(x_i) = \true$ (respectively, $\pi(x_i) = \false$) and hence the clause $C_j$ is satisfied. This proves that the assignment on $\psi$ is a satisfying one and thus, proves the reverse direction of the claim.
\end{proof}


We are now in a position to prove the part of Theorem~\ref{thm:TW}
that states \TCPB\ does not admit an algorithm 
running in time  $2^{2^{o(\tw)}} \cdot {\rm poly}(n)$, unless the \ETH\ fails.

\begin{proof}[Proof of the \TCPB part of Theorem~\ref{thm:TW}]
Assume that there is an algorithm $\calA$ that, given a reduced instance 
$(G, \langle R,B \rangle, k)$ of \TCPB, runs in time $2^{2^{o(\tw)}} \cdot n^{\calO(1)}$
and correctly determines whether it is \yes-instance.
Consider the following algorithm that takes as input
an instance $\psi$ of \textsc{$(3, 3)$-SAT} and determines
whether it is a \yes-instance.
It first constructs an equivalent instance $(G, \langle R,B \rangle, k)$ of \TCPB\
as mentioned in this subsection.
Then, it calls algorithm $\calA$ as a subroutine and returns 
the same answer.
The correctness of this algorithm follows from the correctness
of algorithm $\calA$ and of the reduction (Lemma~\ref{lem:3,3-sat-to-TCPB}).
Note that since each component of $G - V$ is of constant order,
the vertex integrity (and thus treewidth) of $G$ is $\calO(|V|)$.
By the asymptotic estimation of the central binomial coefficient, 
$\binom{2p}{p}\sim \frac{4^p}{\sqrt{\pi \cdot p}}$~\cite{Sperner}. 
To get the upper bound of $2p$, we scale down the asymptotic function and have that $4n \leq \frac{4^p}{2^p}=2^p$.
As we choose the smallest possible value of $p$ such that $2^p \geq 4n$, we can choose $p = \log n + 3$. 
Therefore, $p=\calO(\log (n))$. 
And hence, $|V| = \calO(\log (n))$ which implies $\tw(G) = \calO (\log n)$.
As the other steps, including the construction of the instance 
of \TCPB\, can be performed in the polynomial step,
the running time of the algorithm for \textsc{$(3, 3)$-SAT}
is $2^{2^{o(\log (n))}} \cdot n^{\calO(1)} = 2^{o(n)} \cdot n^{\calO(1)}$.
This, however, contradicts Proposition~\ref{prop:3-3-SAT-ETH-lb}.
Hence, our assumption is wrong, and
\TCPB\ does not admit an algorithm running in time
$2^{2^{o(tw)}} \cdot |V(G)|^{\calO(1)}$, 
unless the {\ETH} fails.
\end{proof}

\subsection{Parameterization by the Solution Size}
\label{subsection:Test Cover para by sol size}

In this subsection, we present a reduction that 
is very close to the reduction used in the proof of Theorem~\ref{thm:locating-dom-set-sol-size-lb}
to prove Theorem~\ref{thm:TC-solsize}, which is restated here for convenience. 

\testcoversolsize*

\subparagraph*{Reduction.}
For notational convenience, instead of specifying an instance
of \TCPB, we specify the incidence graph as mentioned 
in the beginnign of the section.
The reduction takes as input 
an instance $\psi$, with $n$ variables and $m$ clauses, of \textsc{3-SAT} and returns a reduced instance
$(G, \langle R, B \rangle, k)$ of \TCPB
and $k = \calO(\log(n) + \log(m)) = \calO(\log(n))$,
using the sparsification lemma~\cite{DBLP:journals/jcss/ImpagliazzoPZ01}.

We again make use of bit-representation gadgets of~\cite{FoucaudGK0IST24}, as in Section~\ref{subsec:LD-lower-bound-solution-size}.

The reduction constructs graph $G$ as follows.

\begin{itemize} 
\item 
The reduction adds some dummy variables to ensure that 
$n = 2^{2q}$ for some positive integer $q$ which is a power of $2$.
This ensures that $r = \log_2(n) = 2q$ and $s = \frac{n}{r}$ 
both are integers.
It partitions the variables of $\psi$ into $r$ many \emph{buckets} $X_1,
X_2, \dots ,X_r$ such that each bucket contains $s$ many
variables.  
Let $X_i = \{x_{i, j}\ |\ j \in [s]\}$ for all $i\in [r]$. 

\item For every $X_i$, the reduction constructs a set $A_i$ of $2^s$
  many red vertices, that is, $A_i=\{a_{i, \ell}\mid \ell \in [2^s]\}$.
  Each vertex in $A_i$ corresponds to a unique assignment of the
  variables in $X_i$.  
  Moreover, let $A = \cup_{i=1}^r A_i$.
  
\item Corresponding to each $X_i$, let the reduction add a blue
  vertex $b_i$ and the edges $(b_i, a_{i, \ell})$ for all $i \in [r]$ and
  $\ell \in [2^s]$. 
  Let $B' = \{b_i \mid i \in [r]\}$.
 
 \item 
For every clause $C_j$, the reduction adds a pair of blue vertices
$c^{\circ}_j, c^{\star}_j$.  For a vertex $a_{i, \ell} \in A_i$ with
$i \in [r]$, and $\ell \in [2^s]$, if the assignment corresponding to
vertex $a_{i,\ell}$ satisfies the clause $C_j$, then the reduction
adds the edge $(a_{i,\ell}, c^{\circ}_j)$. Let $C = \{c^\circ_j,
c^\star_j \mid j \in [m] \}$.

\item The reduction adds a bit-representation gadget to locate set $C$.
However, it adds the vertices in such a way that for any pair $c^{\circ}_j, c^{\star}_j$, the supplementary vertices
adjacent to them are identical.
\begin{itemize}
\item The reduction sets $p := \lceil \log(m) \rceil+1$ and
for every $i \in [p]$, it adds two vertices, a red vertex $z_{i, 1}$ and 
a blue vertex $z_{i, 2}$, and the edge $(z_{i, 1},  z_{i, 2})$.
\item For every integer $j \in [m]$, let $\bit(j)$ denote the binary representation of $j$ using $p$ bits.
Connect $c^{\circ}_{j}, c^{\star}_j \in C$ with $z_{i, 1}$ if the $i^{th}$ bit in $\bit(j)$ is $1$.
\item Add two vertices $z_{0, 1}$ and $z_{0, 2}$, and edge $(z_{0, 1}, z_{0, 2})$.
Also make every vertex in $C$ adjacent to $z_{0, 1}$.
Let $Z$ be the collection of all the vertices added in this step, and $\bitrep(C)$, the set of red vertices in $Z$.
\end{itemize}
 
\item The reduction adds an isolated blue vertex $b_0$ (whose purpose is to remain uncovered, thereby forcing all other blue vertices to be coverec by some solution test).   
\end{itemize}

\begin{figure}[t]
\centering
\includegraphics[scale=0.75]{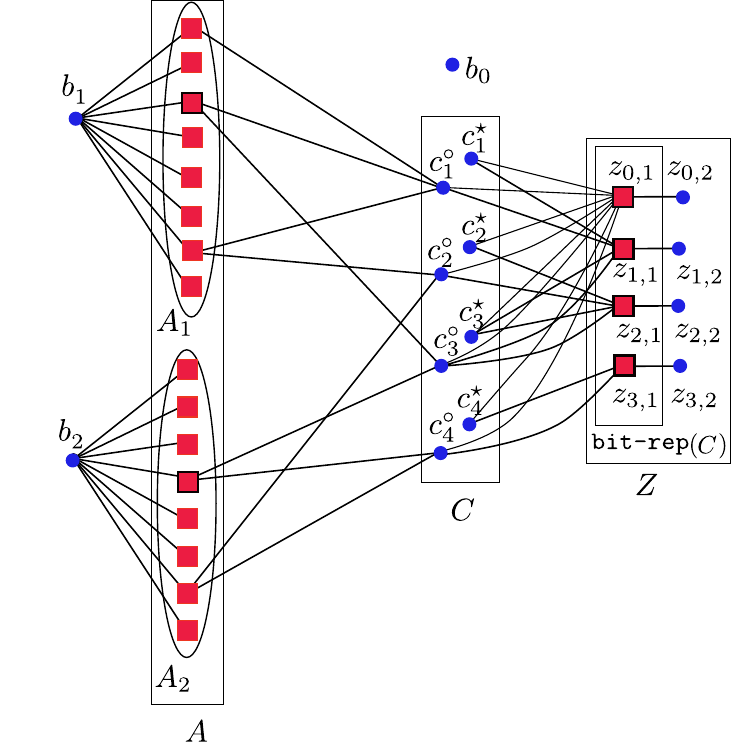}
\caption{An illustrative example of the graph constructed by the reduction in Subsection~\ref{subsection:Test Cover para by sol size}. Red vertices with a thick boundary are part of the solution.
Red (squared) vertices denote the tests whereas blue (filled circle) vertices
the elements. Vertices with a thick boundary are in a potential solution.
\label{fig:test-cover-sol-size}}
\end{figure}

This completes the construction.
We refer to Figure \ref{fig:test-cover-sol-size} for an illustration of the construction.
The reduction sets $k = r + p = 
\calO(\log(n)) + \calO(\log(m)) = \calO(\log(n))$, and
returns $(G, \langle R,B \rangle, k)$ as an instance of \TCPB.

We now prove the correctness of the reduction.


\begin{lemma}\label{lem_TC}
$\psi$ is a \yes-instance of \textsc{$3$-SAT} if and only if
$(G, \langle R, B \rangle, k)$ is a \yes-instance of 
\TCPB.
\end{lemma}
\begin{proof}
Assume first that $\psi$ is a \yes-instance of \textsc{$3$-SAT}. We
construct a subset $R' \subset R$ such that $|R'| \leq k$ in
the following manner. Since $\psi$ is a \yes-instance of $3$-SAT,
there exists a satisfying assignment $\alpha : X \to
\{\true,\false\}$ of $\psi$. Then, the assignment $\alpha$
restricted to each bucket $X_i$ gives an assignment $\alpha_i : X_i
\to \{\true, \false\}$ of the variables in $X_i$. By our construction
of $G$, the assignment $\alpha_i$ corresponds to a particular $a_{i,
  \ell} \in A_i$. We let each $a_{i,\ell}$ corresponding to each
$\alpha_i$ be in $R'$. We complete $R'$ by including in it all
vertices of $\bitrep(C)$. Then, clearly, $R' \subset R$ and $|R'| =
k$.

We first show that $R'$ dominates the set $B \setminus
\{b_0\}$. To start with, the set $\bitrep(C)$ dominates the blue vertices of $Z\cup C$. Moreover, the vertices of $B$ are dominated
by the vertices $a_{i ,\ell}$ picked in $R'$. Thus, $R'$ dominates $B
\setminus \{b_0\}$. We now show that $R'$ locates every pair of
vertices in $B$. Since $b_0$ is the only uncovered vertex, it is located. 
Each blue pendant vertex $z_{i,2}$ has
a unique neighbor $z_{i,1}$ in $R'$ and no other blue vertex is only dominated by $z_{i,1}$ (since all vertices in $C$ are dominated by at least two vertices in $R'$ thanks to $z_{0,1}$). Hence, the vertices of $Z$ are located by $R'$. The pairs $(c^\circ_j, c^\circ_{j'})$,
where $j \neq j'$, are located by $R'$ on account of each such vertex
having a unique neighborhood of size at least~2 in $\bitrep(C)$. The same argument applies to the pairs of the form $(c^\star_j, c^\star_{j'})$ and $(c^\circ_j,
c^\star_{j'})$, where $j \neq j'$. Each vertex $b_i$ of $B'$ is the only one that is adjacent to only one vertex of $R'\cap A$, and all of them are adjacent to a different one, so $B'$ is located by $R'$. Each pair $(b_i, b_{i'})$ is
located by the vertices $a_{i, \ell}$ picked in $R'$ that dominate them. Finally, the pairs
$(c^\circ_j,c^\star_j)$ are located by the vertices in $A\cap R'$, since $\alpha$ is a satisfying assignment and thus eacv vertex $c^\circ_j$ is domonated by at least one vertex of $A\cap R'$. This proves that $R'$, indeed, is a set of tests of the graph $G$.

\smallskip

Let us now assume that $(G,k)$ is a \yes-instance of \TCPB. Therefore,
let $R' \subseteq R$ be a test cover of $G$ such that $|R'| \leq
k$. Since $b_0$ is an isolated vertex of $G$, the set $R'$ dominates
$B\setminus\{b_0\}$ and locates every pair of $B$. Since, for
each $i \in [p]$, the vertex $z_{i,2}$ is pendant and dominated by $R'$,
$R'$ must contain $z_{i,1}$. This implies that $R'$
must contain at most $r$ vertices from $A$. The only vertices of $R'$
that can dominate vertex $b_i$ of $B'$ are those in $A_i$, where $i \in
[r]$. Since there are $r$ such vertices, the set $R'$ must contain at
exactly one vertex $a_{i, \ell}$, say, from  $A_i$, for each $i \in [r]$. 
Since each $a_{i,\ell} \in R'$
corresponds to a partial assignment $\alpha_i : X_i \to \{\true, \false\}$,
therefore, defining $\alpha : X \to \{\true,\false\}$ by $\alpha (x_r)
= \alpha_i (x_r)$ if $x_r \in X_i$ gives an assignment $\alpha$ on
$\psi$. Moreover, $\alpha$ is a satisfying assignment on $\psi$ since,
for each $j \in [m]$, the pair $(c^\circ_j, c^\star_j)$ is located by
some vertex $a_{i,\ell} \in A \cap R'$. Hence, by the construction of
$G$, the corresponding assignment $\alpha_i$ satisfies the clause
$C_j$. This implies that the assignment $\alpha$ satisfies each clause
of $\psi$ and hence, is a satisfying assignment of $\psi$.
\end{proof}

\medskip

We can now prove Theorem~\ref{thm:TC-solsize}, that unless the \ETH\ fails,
  \TCPB\ does not admit an algorithm running in time $2^{2^{o(k)}} \cdot (|U|+|\mathcal F|)^{\calO(1)}$, nor a polynomial-time kernelization algorithm that reduces the solution size and outputs a kernel with $2^{2^{o(k)}}$ vertices.


\begin{proof}[Proof of Theorem~\ref{thm:TC-solsize}]
Assume that there is an algorithm $\calA$ that, given an instance 
$(G,\langle R,B \rangle, k)$ of \TCPB, runs in time $2^{2^{o(k)}} \cdot (|U|+|\calF|)^{\calO(1)}$
and correctly determines whether it is a \yes-instance.
Consider the following algorithm that takes as input
an instance $\psi$ of \textsc{$3$-SAT} and determines
whether it is a \yes-instance.
It first constructs an equivalent instance $(G, \langle R,B \rangle, k)$ of \TCPB\
as mentioned in this subsection.
Note that such an instance can be constructed in time
$2^{\calO(n/\log(n))}$ as the most time consuming step
is to enumerate all the possible assignment of $\calO(n/\log(n))$
many variables in each bucket.

Then, it calls algorithm $\calA$ as a subroutine and returns 
the same answer.
The correctness of this algorithm follows from the correctness
of algorithm $\calA$ and the correctness of the reduction (Lemma~\ref{lem_TC}).
Moreover, from the reduction, we have $k = \calO(\log(n))$.
As the other steps, including the construction of the instance 
of \TCPB, can be performed in time $2^{\calO(n/\log(n))} \cdot n^{\calO(1)} = 2^{o(n)}$,
the running time of the algorithm for \textsc{$3$-SAT}
is $2^{2^{o(\log (n))}} \cdot (|R|+|B|)^{\calO(1)} = 2^{o(n)} \cdot n^{\calO(1)}$.
This, however, contradicts \ETH\ (using the sparsification lemma). 
Hence, our assumption is wrong, and
\TCPB\ does not admit an algorithm running in time
$2^{2^{o(k)}} \cdot (|U| + |\calF|)^{\calO(1)}$, 
unless the {\ETH} fails.

For the second part of Theorem~\ref{thm:TC-solsize}, assume that a kernelization algorithm
$\calB$ exists that takes as input an instance $(G, \langle R,B \rangle, k)$ of \TCPB\
and returns an equivalent instance with $2^{2^{o(k)}}$ vertices.
Then, the brute-force enumerating 
all the possible solutions works in
time $\binom{2^{2^{o(k)}}}{k} \cdot (|R|+|B|)^{\calO(1)}$,
which is $2^{k \cdot 2^{o(k)}} \cdot (|R|+|B|)^{\calO(1)}$,
which is $2^{2^{o(k)}} \cdot (|R|+|B|)^{\calO(1)}$,
contradicting the first result, and hence the \ETH.
Hence, \TCPB does not admit a polynomial-time kernelization algorithm that
returns a kernel with $2^{2^{o(k)}}$ vertices
unless the \ETH\ fails.
\end{proof}

\section{Conclusion}
\label{sec:conclusion}
We presented several results that advance our 
understanding of the algorithmic complexity of \LD and \TCPB, which we showed to have very interesting and rare parameterized complexities.
Moreover, we believe the techniques used in this article
can be applied to other identification problems
to obtain relatively rare conditional lower bounds.
The process of establishing such lower bounds
boils down to designing \emph{bit-representation gadgets}
and \emph{set-representation gadgets} for the problem in question. 

Apart from the broad question of designing such lower bounds for
other identification problems,
we mention an interesting problem left open by our work.
Can our tight double-exponential
lower bound for \LD\ parameterized by vertex integrity (and thus, treedepth, pathwidth and treewidth) be
applied to the feedback vertex set number? In our reductions, after removing a central connecting gadget of logarithmic size, we obtain a collection of small connected components, and thus the graph has small vertex integrity. However, these components are not acyclic, which is why the feedback vertex set number is unbounded. Note that a single-exponential time algorithm with respect to feedback \emph{edge} set number was presented by the authors in~\cite{DBLP:conf/ciac/ChakrabortyFMT25}.

This type of question can of course also be studied for other parameters: see~\cite{CGS21,DBLP:conf/ciac/ChakrabortyFMT25} for works on other structural parameterizations of \LD and \TCPB.

As we have two algorithms for \LD\ running in time $2^{\calO(k^2)}+\calO(k\log n)$ and $2^{2^{\calO(\tw)}}n$ respectively, and they are optimal under \ETH, it would be interesting to determine whether an algorithm running in time, say, $2^{\calO(k+\tw)}n^{\calO(1)}$ exists.

\bibliography{references}

\appendix

\end{document}